\documentclass[english,11pt]{article}
\usepackage[sumlimits]{amsmath}
\usepackage{amsfonts,amssymb,dsfont}
\usepackage{times}
\usepackage{type1cm,graphicx}
\usepackage[usenames]{color}
\usepackage{bbm,mathtools}
\usepackage[left=2cm,top=2cm,right=2cm]{geometry}
\usepackage{hyperref}

\newenvironment{proof}{\noindent\textit{Proof: }}{$\Box $}

\newcommand{\BQP}{{\sffamily BQP}}

\newcommand{\Cl}[1]{Claim~\ref{#1}}

\newcommand{\EqDef}{\stackrel{\mathrm{def}}{=}}

\newcommand{\bra}[1]{\left< #1\right|}
\newcommand{\ket}[1]{\left| #1\right>}

\newcommand{\tr}{\mbox{Tr}}

\newcommand{\mZ}{\mathbbm{Z}}
\newcommand{\mN}{\mathbbm{N}}
\newcommand{\mR}{\mathbbm{R}}
\newcommand{\mbP}{\mathbbm{P}}

\newcommand{\mcI}{\mathcal{I}}

\newcommand{\mfC}{\mathfrak{C}}

\newcommand{\ignore}[1]{}

\newtheorem{thm}{Theorem}[section]

\newtheorem{deff}[thm]{Definition}

\newtheorem{claim}[thm]{Claim}
\newtheorem{lem}[thm]{Lemma}

\newtheorem{scheme}{Scheme}[section]

\newenvironment{proofof}[1]{\noindent{\textit Proof} of $\mathbf{#1}$:\hspace*{1em}}{$\Box $}
\newenvironment{statement}[1]{\noindent{\textbf {#1}\hspace*{0.5em}}}{$\\$}

\def\hpic #1 #2 {\mbox{$\begin{array}[c]{l}
      \epsfig{file=#1,height=#2} \end{array}$}}

\def\vpic #1 #2 {\mbox{$\begin{array}[c]{l}
      \epsfig{file=#1,width=#2} \end{array}$}}

\newcommand{\trnq}[1]{\left[ {#1} \right]_q}

\newcommand{\lwe}{\mathrm{LWE}}

\newcommand{\vc}[1]{\mathbf{{#1}}}
\newcommand{\abs}[1]{\left\vert {#1} \right\vert}

\newcommand{\sivp}{\mathrm{SIVP}}
\newcommand{\otild}{{\widetilde{O}}}

\newcommand{\T}[2]{\left\|#1 - #2\right\|_{tr}}
\def\*#1{\mathbf{#1}}

\newcommand{\GenTrap}{\textsc{GenTrap}}
\newcommand{\Invert}{\textsc{Invert}}

\newcommand{\sR}{{\mathcal{R}}}

\newcommand{\sY}{{\mathcal{Y}}}
\newcommand{\odots}{{\otimes{\ldots}\otimes}}
\newcommand{\cnot}{encrypted CNOT operation}
\newcommand{\capitalcnot}{Encrypted CNOT Operation}

\newcommand{\TV}[2]{\left\|#1 - #2\right\|_{TV}}
\def\*#1{\mathbf{#1}}

\newcommand{\supp}{\textsc{Supp}}

\begin{document}

\author{Urmila Mahadev\footnote{Department of Computer Science, UC Berkeley, USA. Supported by Templeton Foundation Grant 52536, ARO Grant W911NF-12-1-0541, NSF Grant CCF-1410022 and MURI Grant FA9550-18-1-0161. Email: mahadev@cs.berkeley.edu.}}

\title{Classical Homomorphic Encryption for Quantum Circuits}

\maketitle

\begin{abstract}
We present the first leveled fully homomorphic encryption scheme for quantum circuits with classical keys. The scheme allows a classical client to blindly delegate a quantum computation to a quantum server: an honest server is able to run the computation while a malicious server is unable to learn any information about the computation. We show that it is possible to construct such a scheme directly from a quantum secure classical homomorphic encryption scheme with certain properties. Finally, we show that a classical homomorphic encryption scheme with the required properties can be constructed from the learning with errors problem.
\end{abstract}

\section{Introduction}
Can a classical client delegate a desired quantum computation to a remote quantum server while hiding all data from the server? Quantum secure classical encryption schemes do not immediately answer this question; they provide a way of hiding data, but not of computing on the data. This question is particularly relevant to proposals for quantum computing in the cloud. 

The classical analogue of this task, in which the client is a weak classical machine and the server is more powerful (but still classical), was solved in 2009 with the celebrated construction of homomorphic encryption schemes (\cite{homomorphic}). Unfortunately, these schemes are built only to handle classical computations on the encrypted data; the prospect of applying computations in superposition over encrypted bits seems to be much more difficult. This difficulty arises from the fact that all classical homomorphic encryption schemes require (for security) that each bit has many possible different encryptions. This property appears to preclude the quantum property of interference: interference requires that elements of a superposition representing the same bit string, but with opposite amplitudes, cancel out. If these elements have different encryptions, interference cannot happen, thereby preventing one of the key advantages of quantum algorithms.

Due to these obstacles, the question of quantum homomorphic encryption was weakened by allowing a quantum client (\cite{broadbentjeffery2014}). This variant has been well studied in recent years (\cite{broadbentjeffery2014}, \cite{qhomomorphic2}, \cite{qhomomorphic1}, \cite{qhomomorphic3}, \cite{qhomomorphic4}, \cite{speelman2016}) and has led to advancements in the field of delegated quantum computing. However, the model has a number of shortcomings. The principal issue is that the quantum client relies on quantum evaluation keys, which are not reusable; the client must generate fresh keys each time he wishes to delegate a quantum computation. Unfortunately, existing protocols (\cite{speelman2016}) require the client to generate a number of quantum keys proportional to the size of the quantum circuit being applied.\footnote{The client is still restricted in some sense (in comparison to the server); for example, the client may not need to run a general \BQP\ circuit or may only require a constant sized quantum register.} 

The related question of blind quantum computation predated the question of quantum homomorphic encryption and has also been extensively studied, beginning with \cite{childs2001saq}. Blind quantum computation and quantum homomorphic encryption have the same goal of carrying out a computation on encrypted data, but blind computation allows multiple rounds of interaction between the client and server, while homomorphic encryption allows only one round of interaction. Even in the weaker model of blind computation, a quantum client has been a necessity so far (at a minimum, the client must be able to prepare certain constant qubit states \cite{broadbent2008ubq}/\cite{abe2008}). 

In this paper, we return to the original question of quantum homomorphic encryption by providing a homomorphic encryption scheme for quantum computations with a classical client. To do this, we show that certain classical homomorphic encryption schemes can be lifted to the quantum setting; they can be used in a different way to allow for homomorphic evaluation of quantum circuits. It follows that all properties of classical homomorphic encryption schemes (such as reusability of keys and circular security) also hold in the quantum setting. The scheme presented in this paper is the first to allow blind quantum computation between a classical client and a quantum server.\footnote{There have been two papers (\cite{flowpaper}, \cite{kashefipsrqg}) proposing delegated blind quantum computation protocols between a classical client and a quantum server. Both of these results differ from ours in that they do not claim security (i.e. blindness) against a malicious quantum server for the delegation of general quantum computations. Moreover, both results require multiple rounds of interaction between the client and the server.}

To build our homomorphic encryption scheme, we begin with the fact that blindly computing a quantum circuit can be reduced to the ability of a quantum server to perform a CNOT gate (a reversible XOR gate) controlled by classically encrypted data. More specifically, we need a procedure which takes as input the classical encryption of a bit $s$, which we denote as $\textrm{Enc}(s)$, a 2 qubit state $\ket{\psi} = \sum\limits_{a,b\in\{0,1\}}\alpha_{ab}\ket{a,b}$ and outputs the following state:
\begin{equation}\label{eq:introinitialcnot}
    \textrm{CNOT}^s\ket{\psi} = \sum_{a,b\in\{0,1\}}\alpha_{ab}\ket{a,b\oplus a\cdot s}
\end{equation} 
We call this procedure an \cnot. Of course, the output state $\textrm{CNOT}^s\ket{\psi}$ will have to be suitably encrypted (to avoid revealing $s$); we will show in a bit that the Pauli one time pad encryption scheme suffices for this purpose. The key step of the \cnot\ is the extraction of a classically encrypted bit into a quantum superposition. We now show how this extraction can be done by relying on the classical cryptographic primitive of trapdoor claw-free function pairs.

A trapdoor claw-free function pair is a pair of injective functions $f_0,f_1$ which have the same image, are easy to invert with access to a trapdoor, and for which it is computationally difficult to find any pair of preimages $(x_0,x_1)$ with the same image ($f_0(x_0) = f_1(x_1)$). Such a pair of preimages is called a claw, hence the name claw-free. These functions are particularly useful in the quantum setting, due to the fact that a quantum machine can create a uniform superposition over a random claw $(x_0,x_1)$: $\frac{1}{\sqrt{2}}(\ket{x_0} + \ket{x_1})$. This superposition can be used to obtain information which is conjectured to be hard to obtain classically: the quantum machine can obtain \textit{either} a string $d\neq 0$ such that $d\cdot(x_0\oplus x_1) = 0$ or one of the two preimages $x_0,x_1$. In \cite{oneproverrandomness}, this advantage was introduced and used to create a scheme for verifiable quantum supremacy and to show how to generate information theoretic randomness from a single quantum device. This advantage was further used in \cite{oneprover} to classically verify general quantum computations. 

Here, we show that if a bit $s$ is encrypted using a trapdoor claw-free function pair $f_0,f_1:\{0,1\}\times \mathcal{R}\rightarrow \mathcal{Y}$, it can easily be stored in superposition, as required for the \cnot. First, we need to be precise about the encryption: we say that the function pair $f_0,f_1$ encrypts the bit $s$ if each claw $(\mu_0,r_0),(\mu_1,r_1)$ (for $\mu_0,\mu_1\in\{0,1\}, r_0,r_1\in\mathcal{R}$) hides the bit $s$ as follows: $\mu_0\oplus \mu_1 = s$. Before proceeding, we make a simplifying assumption for the purposes of this introduction: we assume that the second qubit of the state $\ket{\psi}$ is fixed to 0 (in other words, $\ket{\psi}$ can be written as $\sum_{a\in\{0,1\}}\alpha_a\ket{a,0}$). This assumption allows us to highlight the key step of the \cnot, which is the extraction of an encrypted bit into a superposition, and allows us to replace \eqref{eq:introinitialcnot} with:
\begin{equation}\label{eq:introsimplecnot}
    \textrm{CNOT}^s\ket{\psi} = \sum_{a\in\{0,1\}}\alpha_{a}\ket{a,a\cdot s}
\end{equation} 
Given the special form of the trapdoor claw-free encryption, conversion to a superposition is quite straightforward. Using the encryption $f_0,f_1$, the prover can entangle the state $\ket{\psi}$ with a random claw, creating the following state:
\begin{equation}\label{eq:introfinalcnot}
    \sum_{a\in\{0,1\}}\alpha_{a}\ket{a}\ket{\mu_a,r_a} = \sum_{a\in\{0,1\}}\alpha_{a}\ket{a}\ket{\mu_0\oplus a\cdot s,r_a}
\end{equation}
Observe that the bit $s$ is now stored in superposition, although it is hidden (due to $\mu_0,r_0,r_1$). As mentioned earlier, the output of the \cnot\ must be suitably encrypted, and $\mu_0,r_0,r_1$ will eventually be part of this encryption, as we will see shortly. To briefly give some intuition about these parameters, note that $\mu_0$ serves as a classical one time pad to hide the bit $s$, and $r_0,r_1$ are required since the bit $s$ is in superposition; therefore, just a classical one time pad does not suffice. 

We now return to the encryption scheme we will use: as in previous blind quantum computation schemes, our homomorphic encryption scheme relies on the Pauli one time pad encryption scheme for quantum states (\cite{pauliotp}), which is the quantum analogue of the classical one time pad. Recall the one time pad method of classical encryption: to encrypt a string $m$, it is XORed with a random string $r$. Just as $l$ classical bits suffice to hide an $l$ bit string $m$, $2l$ classical bits suffice to hide an $l$ bit quantum state $\ket{\psi}$. The Pauli one time pad requires the Pauli $X$ and $Z$ operators, defined as follows:
\begin{equation}
X = \left(\begin{array}{ll}
0&1\\1&0\end{array}\right), Z=\left(\begin{array}{ll}
1&0\\0&-1\end{array}\right)
\end{equation}
A single qubit state $\ket{\psi}$ is information theoretically encrypted by choosing bits $z,x\in\{0,1\}$ (called the Pauli keys) at random and applying the Pauli one time pad $Z^zX^x$, creating the following encryption: 
\begin{equation}\label{eq:intropaulipad}
    Z^zX^x\ket{\psi}
\end{equation}
To convert this encryption to a computationally secure encryption, the encrypted Pauli keys are included as part of the encryption. The resulting encryption is two part, consisting of the quantum state in \eqref{eq:intropaulipad} and classically encrypted bits $\textrm{Enc}(z),\textrm{Enc}(x)$). 

Given the Pauli one time pad, we can now precisely state the goal of our \cnot\ described above. The \cnot\ is a quantum operation which takes as input $\textrm{Enc}(s)$ and a 2 qubit state $\ket{\psi}$ and outputs the state $(Z^z\otimes X^x)\textrm{CNOT}^s\ket{\psi}$ as well as classical encryptions of $z,x\in\{0,1\}$.\footnote{Each application of the \cnot\ results in $z,x\in\{0,1\}^2$ sampled uniformly at random.} Recall the final state of the \cnot\ (as given in \eqref{eq:introfinalcnot}):
\begin{equation}
  \sum_{a\in\{0,1\}}\alpha_{a}\ket{a}\ket{\mu_0\oplus a\cdot s,r_a} =  (\mcI\otimes X^{\mu_0}\otimes\mcI)\sum_{a\in\{0,1\}}\alpha_{a}\ket{a}\ket{a\cdot s,r_a}
\end{equation}
The result of performing a Hadamard measurement on the final register (containing $r_a$) to obtain a string $d$ is:
\begin{equation}
    (Z^{d\cdot(r_0\oplus r_1)}\otimes X^{\mu_0})\sum_{a\in\{0,1\}}\alpha_{a}\ket{a}\ket{a\cdot s}
\end{equation}
To complete the operation, the server must compute the encryptions of $d\cdot(r_0\oplus r_1)$ and $\mu_0$. This can be done via a classical homomorphic computation as long as the server is given the encryption of the trapdoor of the function pair $f_0,f_1$.

So far, we have shown that if $\textrm{Enc}(s)$ is equal to a pair of trapdoor claw-free functions $f_0,f_1$ (for which each claw hides the bit $s$), the \cnot\ is straightforward. However, in order to apply quantum operations homomorphically to the Pauli one time pad encryption scheme described above, we require that $\textrm{Enc}(s)$ is an encryption under a classical homomorphic encryption scheme. Therefore, to obtain a homomorphic encryption scheme for quantum circuits, the above idea must be generalized: we must show how to perform the \cnot\ if $\textrm{Enc}(s)$ is a ciphertext of a classical homomorphic encryption scheme which may not have the trapdoor claw-free structure described in the previous paragraph. We do so by showing that if the classical encryption scheme satisfies certain properties, the function pair $f_0,f_1$ can be constructed given $\textrm{Enc}(s)$. We call such classical encryption schemes \textit{quantum-capable}. The two main results of this paper are that a quantum-capable scheme can be constructed from the learning with errors problem by combining two existing classical encryption schemes (Theorem \ref{thm:dualhequantumcapable}) and that quantum-capable schemes can be used for homomorphic evaluation of quantum circuits (Theorem \ref{thm:qhehomomorphic}). Combined, they provide the following theorem:
\begin{thm}[Informal]\label{thm:mainresult}
Under the assumption that the learning with errors problem with superpolynomial noise ratio is computationally intractable for an efficient quantum machine, there exists a quantum leveled fully homomorphic encryption scheme with classical keys.
\end{thm}
We note that follow up work \cite{BrakerskiQFHE} has improved the underlying assumption to the learning with errors problem with polynomial noise ratio.


\section{Overview}\label{sec:overviewhomomorphic}
We now present an overview of the paper, which proceeds as follows. We first describe the Pauli one time pad encryption scheme and show how quantum gates can be applied homomorphically, with the goal of reducing quantum homomorphic encryption to the \cnot\ described in the introduction. We then describe how the \cnot\ works, in the case that $\textrm{Enc}(s)$ is equal to a trapdoor claw-free function pair $f_0,f_1$ which hides the encrypted bit $s$. Next, we show how a classical encryption of a bit $s$ can be used to build $f_0,f_1$ if the encryption scheme has certain properties and we use these properties to describe how such quantum-capable homomorphic encryption schemes are defined. We conclude by describing how to combine two existing classical homomorphic encryption schemes (\cite{gentry2008}, \cite{fhelwe}) to form a quantum-capable scheme, and then showing how a quantum-capable scheme can be used to build a quantum leveled fully homomorphic encryption scheme with classical keys. The paper itself follows the outline of this overview and provides full proofs of all of our results.

\subsection{Reduction to the Encrypted CNOT Operation}\label{sec:paulipad}

An $l$ qubit quantum state $\ket{\psi}$ is Pauli one time padded by choosing $z,x\in\{0,1\}^l$ at random and applying $Z^zX^x$ to $\ket{\psi}$, creating $Z^zX^x\ket{\psi}$. The bit strings $z,x$ are called the $\textit{Pauli keys}$ and are retained by the client. Once the client sends the encrypted state to the server, the shared state held by the client and server is (the last register containing $zx$ is held by the client):
\begin{equation}
\frac{1}{2^{2l}}\sum_{z,x\in\{0,1\}^l} Z^zX^x\ket{\psi}\bra{\psi} (Z^zX^x)^\dagger\otimes\ket{zx}\bra{zx}
\end{equation}
The client's decoding process is simple: he simply uses the keys $z,x$ to apply the Pauli operator $(Z^zX^x)^\dagger$ to the state he receives from the server. A nice property is that, in the case that $\ket{\psi}$ is a standard basis state $\ket{m}$, the quantum one time pad is the same as the classical one time pad:
\begin{equation}
    Z^zX^x\ket{m}\bra{m}X^xZ^z = \ket{m\oplus x}\bra{m\oplus x}
\end{equation}
It follows that the encoding and decoding of a standard basis state can be performed classically.

The key property of the Pauli one time pad used in blind computing is the fact that it can be used to hide a quantum state entirely: to the server, who has no knowledge of the Pauli keys, a Pauli one time padded quantum state is equal to the maximally mixed state (the identity $\mcI$), as stated in the following lemma (see Section \ref{sec:paulimixingproof} for the proof):
\begin{lem}[\bfseries Pauli Mixing]\label{paulimix}
For a matrix $\rho$ on two spaces $A,B$
$$
\frac{1}{2^{2l}}\sum_{z,x\in\{0,1\}^l} (Z^zX^x\otimes\mcI_B) \rho (Z^zX^x\otimes\mcI_B)^\dagger = \frac{1}{2^l}\mcI_A\otimes\tr_A(\rho)
$$
\end{lem}

The information theoretically secure Pauli one time pad encryption scheme can be easily transformed into a computationally secure encryption scheme (\cite{broadbentjeffery2014}). To do so, the client simply encrypts his classical Pauli keys $(z,x)$ (using a classical homomorphic encryption scheme) and includes the encryption $\textrm{Enc}(z,x)$ as part of the encryption of the state $\ket{\psi}$;  the encryption is now a two part encryption, containing classically encrypted keys and the Pauli one time padded quantum state $Z^zX^x\ket{\psi}$. To decode, the client requests both the Pauli key encryptions and the quantum state. He first decrypts the Pauli key encryptions to obtain the Pauli keys, and then applies the inverse of the Pauli keys to the quantum state.

\subsubsection{Homomorphic Gate Application}
In order to apply quantum gates homomorphically to the computationally secure Pauli encryption scheme, the server will need to run two separate computations: a classical homomorphic computation on the Pauli keys and a quantum computation on the one time padded state. In this section, we show how the server is able to perform the following transformations for all gates $V$ in a universal set of gates (our universal set will be the Clifford group along with the Toffoli gate):
\begin{eqnarray}
\textrm{Enc}(z,x)\rightarrow \textrm{Enc}(z',x')\label{eq:blindcomputingmappingoverview}\\
Z^zX^x\ket{\psi} \rightarrow Z^{z'}X^{x'} V \ket{\psi}\label{eq:blindcomputingmappingoverviewquantum}
\end{eqnarray}

\paragraph{Homomorphic Application of Pauli and Clifford Gates}
To achieve our goal, we simply need to show that the transformations in \eqref{eq:blindcomputingmappingoverview}/\eqref{eq:blindcomputingmappingoverviewquantum} can be computed if $V$ is a Clifford gate or a Toffoli gate. For the case of Clifford gates, we rely on ideas introduced in \cite{childs2001saq} and extended to the computational setting in \cite{broadbentjeffery2014}. To provide intuition, we begin with the case in which $V$ is a Pauli gate $Z^aX^b$. In this case, the server only performs the classical part of the parallel computation (in \eqref{eq:blindcomputingmappingoverview}): he homomorphically updates his Pauli keys from $\textrm{Enc}(z,x)$ to $\textrm{Enc}(z',x')$, for $(z' = z\oplus a, x' = x\oplus b)$. The server has now effectively applied $Z^aX^b$ to his one time padded state, since 
\begin{equation}
   Z^zX^x\ket{\psi} = Z^{z'\oplus a}X^{x'\oplus b}\ket{\psi} = Z^{z'}X^{x'}Z^aX^b\ket{\psi}
\end{equation}
The equality follows up to a global phase since Pauli operators anti commute.

We continue to the case in which $V$ is equal to a Clifford gate $C$. Clifford gates are applied by two parallel computations: a homomorphic Pauli key update and a quantum operation on the one time padded state. The server applies the gate $C$ to his one time padded state, resulting in $CZ^zX^x\ket{\psi}$. We now take advantage of the fact that the Clifford group preserves the Pauli group by conjugation: for all $z,x$, there exist $z',x'$ such that
\begin{equation}
   C Z^zX^x\ket{\psi} = Z^{z'}X^{x'}C\ket{\psi}
\end{equation}
To complete the Clifford application, the server homomorphically updates his Pauli keys from $\textrm{Enc}(z,x)$ to $\textrm{Enc}(z',x')$.

\paragraph{Homomorphic Application of the Toffoli Gate}
The Toffoli gate is more complicated. It cannot be applied by parallel quantum/classical operations by the server, as was done for Clifford gates. This is because it does not preserve Pauli operators by conjugation; applying a Toffoli directly to a 3 qubit one time padded state yields:
\begin{equation}
TZ^zX^x\ket{\psi} = T(Z^zX^x)T^\dagger T\ket{\psi}
\end{equation}
The correction $T(Z^zX^x)T^\dagger$ is not a Pauli operator, as was the case for Clifford operators; it is instead a product of Pauli and Clifford operators, where the Clifford operator involves Hadamard gates and gates of the form $\textrm{CNOT}^{b_{zx}}$ ($b_{zx}$ is a bit which depends on the Pauli keys $z,x$). For the exact form of this operator, see Section \ref{sec:toffoliapp}. Since the correction is a Clifford gate and not a Pauli gate, it cannot be removed by a simple homomorphic Pauli key update by the server. 

In order to complete the application of the Toffoli gate, the server will need to remove the operators $\textrm{CNOT}^{b_{zx}}$ up to Pauli operators. Since the server holds the encrypted Pauli keys, we can assume the server can compute an encryption of $b_{zx}$. Therefore, we have reduced the question of applying a Toffoli gate on top of a one time padded state to the following question: can a \BQP\ server use a ciphertext $c$ encrypting a bit $s$ to apply $\textrm{CNOT}^s$ to a quantum state (up to Pauli operators)? In our setting specifically, $s$ will be a function of the Pauli keys of the one time padded state. 

\subsection{\capitalcnot}\label{sec:descofhiddensum}
We now present the key idea in this paper: we show how a \BQP\ server can apply $\textrm{CNOT}^s$ if he holds a ciphertext $c$ encrypting a bit $s$. We call this procedure an \textit{\cnot}. We first show how to perform this operation in an ideal scenario in which the ciphertext $c$ is a trapdoor claw-free function pair hiding the bit $s$. We then generalize to the case in which $c$ is a ciphertext from a classical homomorphic encryption scheme which satisfies certain properties, which we will describe as they are used. 

In our ideal scenario, there exists finite sets $\sR, \sY$ and the ciphertext $c$ is equal to a trapdoor claw-free function pair $f_0,f_1: \{0,1\}\times\sR\rightarrow\sY$ with one additional property. As a reminder of trapdoor claw-free function pairs, recall that both $f_0,f_1$ are injective and their images are equal. There also exists a trapdoor which allows for efficient inversion of both functions. Note that we have introduced an extra bit in the domain; this bit of the preimage will be used to hide the bit $s$. The property we require is as follows: for all $\mu_0,\mu_1\in\{0,1\}$ and $r_0,r_1\in\sR$ for which $f_0(\mu_0,r_0) = f_1(\mu_1,r_1)$, $\mu_0\oplus \mu_1 = s$ ($s$ is the value encrypted in the ciphertext $c$). 

Our \cnot\ boils down to the ability to extract an encrypted bit from a classical encryption and instead store it in superposition in a quantum state. The claw-free function pair $f_0,f_1$ described above serves as a classical encryption for the bit $s$ which immediately allows this task. At a high level (we discuss the details in the next paragraph), this is because it is possible for the server to compute the following superposition, for a random claw $(\mu_0,r_0),(\mu_1,r_1)$:
\begin{equation}\label{eq:tcfsuperposition0}
    \frac{1}{\sqrt{2}}\sum_{b\in\{0,1\}}\ket{\mu_b}\ket{r_b}
\end{equation}
Since $\mu_0\oplus \mu_1 = s$, the above state can be written as:
\begin{equation}\label{eq:tcftosuperposition}
    (X^{\mu_0}\otimes \mcI)\frac{1}{\sqrt{2}}\sum_{b\in\{0,1\}}\ket{b\cdot s}\ket{r_b}
\end{equation}
Therefore, the server is able to easily convert a classical encryption of $s$ (which is in the form a trapdoor claw-free function pair) to a quantum superposition which contains $s$.

It is quite straightforward to use the above process to apply the \cnot: the superposition over the claw in \eqref{eq:tcfsuperposition0} is simply entangled with the first qubit of the quantum state on which the CNOT is to be applied. We now describe this process in detail. Assume the server would like to apply $\textrm{CNOT}^s$ to a 2 qubit state $\ket{\psi} = \sum_{a,b\in\{0,1\}}\alpha_{ab}\ket{a,b}$. The server begins by entangling the first qubit of the state $\ket{\psi}$ with a random claw of $f_0,f_1$, which proceeds as follows. The server uses the first qubit of $\ket{\psi}$ to choose between the functions $f_0,f_1$ in order to create the following superposition:
\begin{equation}\label{eq:superpositionoverf}
   \frac{1}{\sqrt{2|\sR|}} \sum_{a,b,\mu\in\{0,1\},r\in\sR}\alpha_{ab}\ket{a,b}\ket{\mu,r}\ket{f_a(\mu,r)}
\end{equation}
Now the server measures the final register to obtain $y\in \sY$. Let $(\mu_0,r_0), (\mu_1,r_1)$ be the two preimages of $y$ ($f_0(\mu_0,r_0) = f_1(\mu_1,r_1) = y$). The remaining state is:
\begin{equation}\label{eq:statecollapseoverview}
    \sum_{a,b\in\{0,1\}}\alpha_{ab}\ket{a,b}\ket{\mu_a}\ket{r_a}
\end{equation}
Recall that to apply CNOT$^s$, the value $a\cdot s$ must be added to the register containing $b$. This is where the structure in \eqref{eq:tcftosuperposition} (which relies on the fact that $\mu_0\oplus \mu_1 = s$) comes in to play: to add $a\cdot s$, the server XORs $\mu_a$ into the second register, which essentially applies the operation CNOT$^s$:
\begin{equation}
    \sum_{a,b\in\{0,1\}}\alpha_{ab}\ket{a,b\oplus \mu_a}\ket{\mu_a}\ket{r_a} = \sum_{a,b\in\{0,1\}}\alpha_{ab}(\mcI\otimes X^{\mu_0})\textrm{CNOT}_{1,2}^s\ket{a,b}\otimes \ket{\mu_a,r_a} 
\end{equation}
Finally, the server applies a Hadamard transform on the registers containing $\mu_a,r_a$ and measures to obtain $d$. If we let $(\mu_a,r_a)$ denote the concatenation of the two values, the resulting state (up to a global phase) is
\begin{equation}
(Z^{d\cdot((\mu_0,r_0)\oplus (\mu_1,r_1))}\otimes X^{\mu_0})\textrm{CNOT}_{1,2}^s\sum_{a,b\in\{0,1\}}\alpha_{ab}\ket{a,b}
\end{equation}
In order to complete the \cnot, the server requires an encryption of the trapdoor of the functions $f_0,f_1$. The server can then homomorphically compute the bits $\mu_0$ and $d\cdot((\mu_0,r_0)\oplus (\mu_1,r_1))$ and use these bits to update his Pauli keys.

\subsubsection{Trapdoor Claw-free Pair Construction}
So far, we have shown how to apply the \cnot\ in the case that the ciphertext $c$ encrypting the bit $s$ is a trapdoor claw-free function pair $f_0,f_1$ which hides $s$. To build a homomorphic encryption scheme, we need to show how the \cnot\ can be applied if $c$ instead comes from a classical homomorphic encryption scheme, which we call HE. We now show that if HE satisfies certain properties, the function pair $f_0,f_1$ hiding the bit $s$ can be constructed (by the server) using $c$. 

The function $f_0$ will be the encryption function of HE. The function $f_1$ is the function $f_0$ shifted by the homomorphic XOR of the ciphertext $c$ encrypting the bit $s$: $f_1 = f_0 \oplus_H c$ ($\oplus_H$ is the homomorphic XOR operation). To ensure that $f_0,f_1$ are injective, we require that HE has the property of randomness recoverability: there must exist a trapdoor which allows recovery of $\mu_0,r_0$ from a ciphertext $\textrm{Enc}(\mu_0;r_0)$ ($\textnormal{Enc}(\mu_0;r_0)$ denotes the encryption of a bit $\mu_0$ with randomness $r_0$). We also require that the homomorphic XOR operation is efficiently invertible using only the public key of HE.

Unfortunately, the images of the functions $f_0,f_1$ are not equal. We will instead require the weaker (but still sufficient) condition that there exists a distribution $D$ over the domain of the functions such that $f_0(D)$ and $f_1(D)$ are statistically close. We replace \eqref{eq:superpositionoverf} with the corresponding weighted superposition:
\begin{equation}\label{eq:overviewsuperpositionoverD}
    \sum_{a,b,\mu\in\{0,1\},r}\alpha_{ab}\sqrt{D(\mu,r)}\ket{a}\ket{b}\ket{\mu,r}\ket{f_a(\mu,r)}
\end{equation}
To do so, we require that the server can efficiently create the following superposition: 
\begin{equation}\label{eq:distributionsuperposition}
    \sum_{\mu\in\{0,1\},r}\sqrt{D(\mu,r)}\ket{\mu,r}
\end{equation}
Due to the negligible statistical distance between $f_0(D)$ and $f_1(D)$, when the last register of \eqref{eq:overviewsuperpositionoverD} is measured to obtain $y$, with high probability there exist $\mu_0,r_0,\mu_1,r_1$ such that $y = f_0(\mu_0,r_0) = f_1(\mu_1,r_1)$, which implies that the state collapses to \eqref{eq:statecollapseoverview} and that 
\begin{equation}
    \textrm{Enc}(\mu_0;r_0) = \textrm{Enc}(\mu_1;r_1) \oplus_H c
\end{equation}
Since $\oplus_H$ is the homomorphic XOR operation, $\mu_0\oplus \mu_1 = s$. 

There is one remaining issue with the \cnot\ described above. The requirements above must hold for a classical ciphertext $c$ which occurs at any point during the classical computation on the encrypted Pauli keys. However, in many classical homomorphic encryption schemes, the format of the ciphertext $c$ changes throughout the computation. We know of several schemes for which the above requirements hold for a freshly encrypted ciphertext, but we do not know of any schemes which satisfy the requirements during a later stage of computation. The next section addresses this complication by sufficiently weakening the above requirements while preserving the functionality of the \cnot.

\subsection{Quantum-Capable Classical Homomorphic Encryption Schemes}\label{sec:quantumcapableoverview}
In this section, we define quantum-capable homomorphic encryption schemes, i.e. classical leveled fully homomorphic encryption schemes which can be used to evaluate quantum circuits. To justify why we must weaken the requirements listed in Section \ref{sec:descofhiddensum}, we begin with a description of the ideal high level structure of a quantum-capable homomorphic encryption scheme. In many classical homomorphic encryption schemes, the encryption of a bit $b$ can be thought of as a random element of a subset $S_b$, perturbed by some noise term $\epsilon$. As the computation progresses, we will require that the structure of the ciphertext remains the same; it must still be a random element of $S_b$, but the noise term may grow throughout the computation. We will also require that the homomorphic XOR operation is natural, in the sense that the noise of the output ciphertext is simply the addition of the two noise terms of the input ciphertexts. If these two conditions hold (invariance of the ciphertext form and the existence of a natural XOR operation), deriving a distribution $f_0(D)$ over ciphertexts which remains roughly the same after shifting by the homomorphic XOR of the ciphertext $c$ (as needed in Section \ref{sec:descofhiddensum}) is straightforward. We simply choose $D$ to sample the noise term from a discrete Gaussian distribution with width sufficiently larger than the magnitude of the noise term of the ciphertext $c$. 

Unfortunately, we do not know of a classical homomorphic encryption scheme which satisfies both the conditions (ciphertext form and natural XOR operation) at once. To account for this difficulty, we define a quantum-capable homomorphic encryption schemes as follows. We call a classical homomorphic encryption scheme HE quantum-capable if there exists an alternative encryption scheme AltHE which satisfies the following conditions. First, given a ciphertext $c$ under HE, it must be possible for the server to convert $c$ to a ciphertext $\hat{c}$ under AltHE. The conversion process must maintain the decrypted value of the ciphertext. Second, AltHE must have a natural homomorphic XOR operation (which is also efficiently invertible). Third, there must exist a distribution $f_0(D)$ over encryptions under AltHE which must remain almost the same after shifting by the homomorphic XOR of $\hat{c}$ and must allow efficient construction of the superposition in \eqref{eq:distributionsuperposition}. In addition, it must be possible to both decrypt and recover randomness from ciphertexts under AltHE given the appropriate secret key and trapdoor information. This definition is formalized in Section \ref{sec:quantumcapablerequirements}.

Finally, we describe how to connect this weaker definition to the \cnot\ given in Section \ref{sec:descofhiddensum}. We begin with a quantum-capable homomorphic encryption scheme HE, which is used to encrypt the Pauli keys. HE satisfies the ciphertext form requirement but may not have a natural XOR operation. Each time the server needs to apply an \cnot\ (controlled by a ciphertext $c$ encrypting a bit $s$ under HE), he will convert $c$ to a ciphertext $\hat{c}$ under AltHE, which does have a natural XOR operation. Using AltHE (rather than HE) the server performs the operations described in Section \ref{sec:descofhiddensum}. Upon obtaining his measurement results (denoted as $y$ and $d$), the server will encrypt both $\hat{c}$ and $y,d$ under HE. The server will then use the secret key and trapdoor information of AltHE, which are provided to him as encryptions under HE, to homomorphically recover the randomness and decrypted values from both $y$ and $\hat{c}$. This encrypted information can be used to homomorphically compute the Pauli key updates. The entire classical homomorphic computation is done under HE.

\subsection{Example of a Quantum-Capable Classical Encryption Scheme}\label{sec:overviewquantumcapableexample}
In Section \ref{sec:quantumcapableexample} (Theorem \ref{thm:dualhequantumcapable}), we show that an existing classical fully homomorphic encryption scheme is quantum-capable. We use the structure of the scheme from \cite{fhelwe}, which is a leveled fully homomorphic encryption scheme built by extending the vector ciphertexts of \cite{regev2005} to matrices. The resulting encryption scheme does satisfy the ciphertext form requirement (as described in Section \ref{sec:quantumcapableoverview}), but since the underlying encryption scheme (\cite{regev2005}) does not have the randomness recoverability property, neither does \cite{fhelwe}. We therefore alter the scheme from \cite{fhelwe} to use the dual encryption scheme of \cite{regev2005}, which was introduced in \cite{gentry2008} and allows randomness recovery, as the underlying encryption scheme. We call the resulting scheme DualHE and the underlying scheme of \cite{gentry2008} Dual. 

We use the scheme DualHE as an instantiation of the scheme we called HE in Section \ref{sec:quantumcapableoverview}. Although the underlying scheme Dual does have a natural XOR operation, the extension to matrices compromises the XOR operation; once the ciphertexts are matrices, addition is performed over a larger field. Luckily, it is easy to convert a ciphertext of DualHE to a ciphertext of Dual. We therefore use Dual as AltHE. 

In Section \ref{sec:quantumcapableexample}, we first describe the scheme Dual from \cite{gentry2008} and we then show how to extend it to DualHE using \cite{fhelwe}. Next, we show that DualHE satisfies the ciphertext form requirement and that a ciphertext under DualHE can be converted to a ciphertext under Dual. In Theorem \ref{thm:dualhehomomorphic}, we use these properties to show  that DualHE is a classical leveled fully homomorphic encryption scheme. Finally, we show in Theorem \ref{thm:dualhequantumcapable} that DualHE is quantum-capable with only a small modification of parameters (the underlying assumption for both the classical FHE and the quantum-capable instantiation is the hardness of learning with errors with a superpolynomial noise ratio). 

\subsection{Extension to Quantum Leveled Fully Homomorphic Encryption}\label{sec:extensiontohomomorphicoverview}
We have so far provided a quantum fully homomorphic encryption scheme with classical keys under the assumption of circular security: the server must be provided the encrypted secret key and trapdoor information in order to update his encrypted Pauli keys after each \cnot. Our notion of circular security here will be slightly stronger than the standard notion, due to the encryption of the trapdoor (instead of just the secret key). As an alternative to assuming circular security, we can build a quantum leveled fully homomorphic encryption scheme by employing a technique which is commonly used in classical homomorphic encryption schemes (Section 4.1 in \cite{homomorphic}): we will encrypt the secret key and trapdoor information under a fresh public key. In other words, the $i^{th}$ level secret key $sk_i$ and its corresponding trapdoor information are encrypted under a fresh public key $pk_{i+1}$ and given to the server as part of the evaluation key. The computation of the Pauli key updates corresponding to the \cnot s of level $i$ is performed under $pk_{i+1}$ (i.e. the corresponding $\hat{c},y$ and $d$ from each \cnot\ in level $i$ will be encrypted, by the server, under $pk_{i+1}$ - see the last paragraph of Section \ref{sec:quantumcapableoverview}). 

Note that with the introduction of the leveled scheme, we can see the classical portion of the quantum homomorphic computation as follows. Each level of the quantum computation can be thought of as a series of Clifford gates followed by one layer of non intersecting Toffoli gates, finishing with a layer of non intersecting \cnot s. It follows that the classical homomorphic computation of level $i$ consists of first decrypting and recovering randomness from the cipertexts corresponding to the \cnot s from level $i - 1$, then performing the Pauli key updates corresponding to the \cnot s from level $i - 1$ and finally performing the Pauli key updates corresponding to the Clifford and Toffoli gates of level $i$. The ciphertexts which result from this computation are then used as the control bits for the layer of \cnot s of level $i$. 

Intuitively, this leveled approach is secure since each secret key is protected by the semantic security of the encryption scheme under an independent public key. To prove security, we start with the final level of encryption. If there are $L$ levels of the circuit, then there will be no trapdoor or secret key information provided corresponding to $pk_{L+1}$. It follows that all encryptions under $pk_{L+1}$ can be replaced by encryptions of 0; now there is no encrypted information provided corresponding to $sk_L$. Then all encryptions under $pk_L$ can be replaced by encryptions of 0, and we can continue in this manner until we reach $pk_1$, which will imply security of encryptions under the initial public key $pk_1$. The scheme and proof of security are presented in Section \ref{sec:extensiontohomomorphic}, proving that quantum-capable classical encryption schemes can be used to build a quantum leveled fully homomorphic encryption scheme (see Theorem \ref{thm:qhehomomorphic} for a formal statement).

\subsection{Paper Outline} 
We begin with preliminaries in Section \ref{sec:homomorphicprelim0}. In Section \ref{sec:quantumcapablerequirements}, we define quantum-capable encryption schemes by listing the requirements that a classical homomorphic encryption scheme must satisfy in order to be used to evaluate quantum circuits, as described in Section \ref{sec:quantumcapableoverview}. We use this definition to formally prove the correctness (in Claim \ref{cl:sumcorrectness}) of the \cnot\ given in Section \ref{sec:descofhiddensum}. Section \ref{sec:quantumcapableexample} covers Section \ref{sec:overviewquantumcapableexample}: we provide an example of a classical homomorphic encryption scheme which is quantum-capable. In Section \ref{sec:extensiontohomomorphic}, we formally show how to extend a quantum-capable classical leveled fully homomorphic encryption scheme to a quantum leveled fully homomorphic encryption scheme (as described in Section \ref{sec:extensiontohomomorphicoverview}).

\section{Preliminaries}\label{sec:homomorphicprelim0}

Throughout this paper, we borrow notation and definitions from \cite{abe2008}, \cite{abem} and \cite{oneproverrandomness}. Parts of the following sections are also taken from these sources. 

\subsection{Notation}\label{sec:prelimnotation}
For all $q \in \mN$ we let $\mZ_q$ denote the ring of integers modulo $q$. We represent elements in $\mZ_q$ using numbers in the range $(-\frac{q}{2}, \frac{q}{2}] \cap \mZ$. We denote by $\trnq{x}$ the unique integer $y$ s.t.\ $y = x \pmod{q}$ and $y \in (-\frac{q}{2}, \frac{q}{2}]$. For $x\in\mZ_q$ we define $\abs{x}=|{\trnq{x}}|$. For a vector $\*u\in \mZ_q^n$, we write $\lVert \*u \rVert_{\infty}\leq \beta$ if each entry $u_i$ in $\*u$ satisfies $|u_i|\leq \beta$. Similarly, for a matrix $U\in \mZ_q^{n\times m}$, we write $\lVert U \rVert_{\infty}\leq \beta$ if each entry $u_{i,j}$ in $U$ satisfies $|u_{i,j}|\leq \beta$. When considering an $s\in \{0,1\}^n$ we sometimes also think of $s$ as an element of $\mZ_q^n$, in which case we write it as $\*s$. 

We use the terminology of polynomially bounded, super-polynomial, and negligible functions. A function $n: \mN \to \mR_+$ is \emph{polynomially bounded} if there exists a polynomial $p$ such that $n(\lambda)\leq p(\lambda)$ for all $\lambda \in \mN$. A function $n: \mN \to \mR_+$ is \emph{negligible} (resp. \emph{super-polynomial}) if for every polynomial $p$, $p(\lambda) n(\lambda)\to_{\lambda\to\infty} 0$ (resp. $ n(\lambda)/p(\lambda)\to_{\lambda\to\infty} \infty$).

We generally use the letter $D$ to denote a distribution over a finite domain $X$, and $f$ for a density on $X$, i.e. a function $f:X\to[0,1]$ such that $\sum_{x\in X} f(x)=1$. We often use the distribution and its density interchangeably. We write $U$ for the uniform distribution. We write $x\leftarrow D$ to indicate that $x$ is sampled from distribution $D$, and $x\leftarrow_U X$ to indicate that $x$ is sampled uniformly from the set $X$. 
We write $\mathcal{D}_X$ for the set of all densities on $X$.
For any $f\in\mathcal{D}_X$, $\supp(f)$ denotes the support of $f$,
\begin{equation*}
    \supp(f) \,=\, \big\{x\in X \,|\; f(x)> 0\big\}\;.
\end{equation*}
For two densities $f_1$ and $f_2$ over the same finite domain $X$, the Hellinger distance  between $f_1$ and $f_2$ is
\begin{equation}\label{eq:bhatt}
H^2(f_1,f_2) \,=\, 1- \sum_{x\in X}\sqrt{f_1(x)f_2(x)}\;.
\end{equation}
and the total variation distance between $f_1$ and $f_2$ is:
\begin{equation}\label{eq:deftotalvariation}
\TV{f_1}{f_2} \,=\, \frac{1}{2} \sum_{x\in X}|f_1(x) - f_2(x)|\;.
\end{equation}
The following lemma will be useful:
\begin{lem}\label{lem:projectionprob}
Let $D_0, D_1$ be distributions over a finite domain $X$. Let $X'\subseteq X$. Then:
\begin{equation}
    \Big|\Pr_{x\leftarrow D_0}[x\in X'] - \Pr_{x\leftarrow D_1}[x\in X']\Big| \,\leq\, \TV{D_0}{D_1}
\end{equation}
\end{lem}
We require the following definition:
\begin{deff}{\textbf{Computational Indistinguishability of Distributions}}\label{def:compinddist}
Two families of distributions $\{D_{0,\lambda}\}_{\lambda\in\mN}$ and $\{D_{1,\lambda}\}_{\lambda\in\mN}$ (indexed by the security parameter $\lambda$) are computationally indistinguishable if for all quantum polynomial-time attackers $\mathcal{A}$ there exists a negligible function $\mu(\cdot)$ such that for all $\lambda\in\mN$
\begin{equation}
\Big|\Pr_{x\leftarrow D_{0,\lambda}}[\mathcal{A}(x) = 0] - \Pr_{x\leftarrow D_{1,\lambda}}[\mathcal{A}(x) = 0]\Big| \,\leq\, \mu(\lambda)\;.
\end{equation}
\end{deff}

\subsection{Learning with Errors and Discrete Gaussians}\label{sec:lweprelim}

This background section on the learning with errors problem is taken directly from \cite{oneproverrandomness}. For a positive real $B$ and positive integers $q$, the truncated discrete Gaussian distribution over $\mZ_q$ with parameter $B$ is supported on $\{x\in\mZ_q:\,\|x\|\leq B\}$ and has density
\begin{equation}\label{eq:d-bounded-def}
 D_{\mZ_q,B}(x) \,=\, \frac{e^{\frac{-\pi\lVert x\rVert^2}{B^2}}}{\sum\limits_{x\in\mZ_q,\, \|x\|\leq B}e^{\frac{-\pi\lVert x\rVert^2}{B^2}}} \;.
\end{equation}
We note that for any $B>0$, the truncated and non-truncated distributions have statistical distance that is exponentially small in $B$~\cite[Lemma 1.5]{banaszczyk1993new}. For a positive integer $m$, the truncated discrete Gaussian distribution over $\mZ_q^m$ with parameter $B$ is supported on $\{x\in\mZ_q^m:\,\|x\|\leq B\sqrt{m}\}$ and has density
\begin{equation}\label{eq:d-bounded-def-m}
\forall x = (x_1,\ldots,x_m) \in \mZ_q^m\;,\qquad D_{\mZ_q^m,B}(x) \,=\, D_{\mZ_q,B}(x_1)\cdots D_{\mZ_q,B}(x_m)\;.
\end{equation}

\begin{lem}\label{lem:distributiondistance}
Let $B$ be a positive real number and $q,m$ be positive integers. Let $\*e \in \mZ_q^m$. The Hellinger distance between the distribution $D = D_{\mZ_q^{m},B}$ and the shifted distribution $D+\*e$ satisfies
\begin{equation}
H^2(D,D+\*e) \,\leq\, 1- e^{\frac{-2\pi \sqrt{m}\|\*e\|}{B}}\;,
\end{equation}
and the statistical distance between the two distributions satisfies
\begin{equation}
\big\| D - (D+\*e) \big\|_{TV}^2 \,\leq\, 2\Big(1 - e^{\frac{-2\pi \sqrt{m}\|\*e\|}{B}}\Big)\;.
\end{equation}
\end{lem}

\begin{deff}\label{def:lwehardness}
For a security parameter $\lambda$, let $n,m,q\in \mN$ be integer functions of $\lambda$. Let $\chi = \chi(\lambda)$ be a distribution over $\mZ$. The $\lwe_{n,m,q,\chi}$ problem is to distinguish between the distributions $(\*A, \*A\*s + \*e \pmod{q})$ and $(\*A, \*u)$, where $\*A$ is uniformly random in $\mZ_q^{n \times m}$, $\*s$ is a uniformly random row vector in $\mZ_q^n$, $\vc{e}$ is a  row vector drawn at random from the distribution $\chi^m$, and $\*u$ is a uniformly random vector in $\mZ_q^m$. Often we consider the hardness of solving $\lwe$ for {any} function $m$ such that $m$ is at most a polynomial in $n \log q$. This problem is denoted $\lwe_{n,q,\chi}$. When we write that we make the $\lwe_{n,q,\chi}$ assumption, our assumption is that no quantum polynomial-time procedure can solve the $\lwe_{n,q,\chi}$ problem with more than a negligible advantage in $\lambda$. 
\end{deff}

As shown in \cite{regev2005,PRS17}, for any $\alpha>0$ such that  $\sigma = \alpha q \ge 2 \sqrt{n}$ the $\lwe_{n,q,D_{\mZ_q,\sigma}}$ problem,  where $D_{\mZ_q,\sigma}$ is the discrete Gaussian distribution, is at least as hard as approximating the shortest independent vector problem ($\sivp$) to within a factor of $\gamma = \otild({n}/\alpha)$ in \emph{worst case} dimension $n$ lattices. This is proven using a quantum reduction. Classical reductions (to a slightly different problem) exist as well \cite{Peikert09,BLPRS13} but with somewhat worse parameters. The best known (classical or quantum) algorithm for these problems run in time $2^{\otild(n/\log \gamma)}$. For our construction we assume hardness of the problem against a quantum polynomial-time adversary in the case that $\gamma$ is a super polynomial function in $n$. This is a commonly used assumption in cryptography (for e.g. homomorphic encryption schemes such as \cite{fhelwe}).

We use two additional properties of the LWE problem. The first is that it is possible to generate LWE samples $(\*A,\*A\*s+\*e)$ such that there is a trapdoor allowing recovery of $\*s$ from the samples. 

\begin{thm}[Theorem 5.1 in~\cite{miccancio2012}]\label{thm:trapdoor}
Let $n,m\geq 1$ and $q\geq 2$ be such that $m = \Omega(n\log q)$. There is an efficient randomized algorithm $\GenTrap(1^n,1^m,q)$ that returns a matrix $\*A \in \mZ_q^{m\times n}$ and a trapdoor $t_{\*A}$ such that the distribution of $\*A$ is negligibly (in $n$) close to the uniform distribution. Moreover, there is an efficient algorithm $\Invert$ that, on input $\*A, t_{\*A}$ and $\*A\*s+\*e$ where $\|\*e\| \leq q/(C_T\sqrt{n\log q})$ and $C_T$ is a universal constant, returns $\*s$ and $\*e$ with overwhelming probability over $(\*A,t_{\*A})\leftarrow \GenTrap$. \end{thm}

\subsection{Quantum Computation Preliminaries}

\subsubsection{Quantum Operations}\label{sec:PauliNClifford}
We will use the $X, Y$ and $Z$ Pauli operators: $X = \left(\begin{array}{ll}
0&1\\1&0\end{array}\right)$, $Z=\left(\begin{array}{ll}
1&0\\0&-1\end{array}\right)$ and $Y = iXZ$. The $l$-qubits Pauli group consists of all elements of the form
$P=P_1\otimes P_2\odots  P_l$ where $P_i \in \{\mcI,X,Y,Z\}$, together with the multiplicative factors $-1$ and $\pm i$. We will use a subset of this group, which we denote as $\mbP_l$, which includes all operators $P = P_1\otimes P_2\odots P_l$ but not the multiplicative factors. We will use the fact that Pauli operators anti commute; $ZX = - XZ$. The Pauli group $\mbP_l$ is a basis to the matrices acting on $l$ qubits. We can write any matrix $U$ over a vector space $A\otimes B$ (where $A$ is the space of $l$ qubits) as $\sum_{P\in \mbP_l}P\otimes U_P$ where $U_P$ is some (not necessarily unitary) matrix on $B$. 

Let $\mfC_l$ denote the $l$-qubit Clifford group. Recall that it is a finite subgroup of unitaries acting on $l$ qubits generated by the Hadamard matrix $H=\frac{1}{\sqrt{2}}\left(\begin{array}{ll}
1&1\\1&-1\end{array}\right)$, by $K=\left(\begin{array}{ll}
1&0\\0&i\end{array}\right)$, and by controlled-NOT (CNOT) which maps $\ket{a,b}$ to $\ket{a,a\oplus b}$ (for bits $a,b$). We will also use the controlled phase gate, $\hat{Z}$:
\begin{eqnarray}     
\hat{Z}\ket{a,b} &=& (-1)^{ab}\ket{a,b}\\
\hat{Z} &=& (\mcI\otimes H)CNOT(\mcI\otimes H)
\end{eqnarray}
The Clifford group is characterized by the property that it maps the
Pauli group $\mbP_l$ to itself, up to a phase $\alpha\in\{\pm 1,\pm i\}$. That is:
$\forall C\in\mfC_l ,  P\in \mbP_l: ~\alpha CPC^\dagger \in \mbP_l$

The Toffoli gate $T$ maps  $\ket{a,b,c}$ to $\ket{a,b,c\oplus ab}$ (for $a,b,c\in\{0,1\}$). We will use the fact that the Clifford group combined with the Toffoli gate constitute a universal gate set for quantum circuits.

\subsubsection{Toffoli Gate Application}\label{sec:toffoliapp}
A straightforward calculation shows that the Toffoli operator maps Pauli operators to Clifford operators in the following way:
\begin{eqnarray}
TZ^{z_1}X^{x_1}\otimes Z^{z_2}X^{x_2} \otimes Z^{z_3}X^{x_3} T^\dagger
&=& CNOT_{1,3}^{x_2} CNOT_{2,3}^{x_1} \hat{Z}_{1,2}^{z_3}Z^{z_1+x_2z_3}X^{x_1}\otimes Z^{z_2 + x_1z_3}X^{x_2} \otimes Z^{z_3}X^{x_1x_2 + x_3} \nonumber
\end{eqnarray}
We can split the above expression into Clifford and Pauli operators as follows:
\begin{eqnarray}
C_{zx} &=& CNOT_{1,3}^{x_2} CNOT_{2,3}^{x_1} \hat{Z}_{1,2}^{z_3}\nonumber\\
P_{zx} &=& Z^{z_1+x_2z_3}X^{x_1}\otimes Z^{z_2 + x_1z_3}X^{x_2} \otimes Z^{z_3}X^{x_1x_2 + x_3} 
\end{eqnarray}
We therefore obtain:
\begin{eqnarray}
TZ^{z_1}X^{x_1}\otimes Z^{z_2}X^{x_2} \otimes Z^{z_3}X^{x_3} T^\dagger &=& C_{zx}P_{zx}
\end{eqnarray}
Observe that $C_{zx}$ consists only of $\textrm{CNOT}$ gates and 2 Hadamard gates. Moreover, the three Clifford gates in $C_{zx}$ (the two CNOT operators and the controlled phase operator) commute. Also note that only the $\textrm{CNOT}$ gates are dependent on the Pauli keys.

\subsubsection{Pauli Mixing}\label{sec:paulimixingproof}
Here we restate Lemma \ref{paulimix} and prove it:

\begin{statement}{\textbf{Lemma \ref{paulimix}}}
\textit{For a matrix $\rho$ on two spaces $A,B$
$$
\frac{1}{2^{2l}}\sum_{z,x\in\{0,1\}^l} (Z^zX^x\otimes\mcI_B) \rho (Z^zX^x\otimes\mcI_B)^\dagger = \frac{1}{2^l}\mcI_A\otimes\tr_A(\rho)
$$}
\end{statement}
\begin{proofof}{\textbf{Lemma \ref{paulimix} }}
First, we write $\rho$ as:
$$
\sum_{ij}\ket{i}\bra{j}_A\otimes \rho_{ij}
$$
It follows that:
$$
\tr_A(\rho) = \sum_i\rho_{ii}
$$
Next, observe that:
\begin{eqnarray}
\sum_{z,x\in\{0,1\}^l} Z^zX^x\ket{i}\bra{j}(Z^zX^x)^\dagger &=& \sum_{x\in\{0,1\}^l}(\sum_{z\in\{0,1\}^l}(-1)^{z\cdot (i\oplus j)}) X^x\ket{i}\bra{j}(X^x)^\dagger\
\end{eqnarray}
This expression is 0 if $i\neq j$. If $i = j$, we obtain $2^l\mcI_A$. Plugging in this observation to the expression in the claim, we have:
\begin{eqnarray}
\frac{1}{2^{2l}}\sum_{z,x\in\{0,1\}^l} (Z^zX^x\otimes\mcI_B)\rho(Z^zX^x\otimes\mcI_B)^\dagger &=&
\frac{1}{2^{2l}}\sum_{ij}\sum_{z,x\in\{0,1\}^l} Z^zX^x\ket{i}\bra{j}_A (Z^zX^x)^\dagger\otimes \rho_{ij}\nonumber\\
&=& \frac{1}{2^{2l}}\sum_{i}\sum_{z,x\in\{0,1\}^l} Z^zX^x\ket{i}\bra{i}_A (Z^zX^x)^\dagger\otimes \rho_{ii}\nonumber\\
&=&\frac{1}{2^l}\mcI_A\otimes\tr_A(\rho)
\end{eqnarray}
\end{proofof}

\subsubsection{Trace Distance}\label{sec:prelimtrace}
For density matrices $\rho,\sigma$, the trace distance $\T{\rho}{\sigma}$ is equal to:
\begin{eqnarray}\label{eq:deftracedistance}
\T{\rho}{\sigma} = \frac{1}{2}\tr(\sqrt{(\rho - \sigma)^2})
\end{eqnarray}
We will use the following fact (\cite{tracedistance}):
\begin{equation}
\T{\rho}{\sigma} = \max_P \tr(P(\rho - \sigma))
\end{equation}
where the maximization is carried over all projectors $P$. We will also use the fact that the trace distance is contractive under completely positive trace preserving maps (\cite{tracedistance}). The following lemma relates the Hellinger distance as given in \eqref{eq:bhatt} and the trace distance of superpositions:
\begin{lem}\label{lem:hellingertotrace}
Let $X$ be a finite set and $f_1,f_2\in\mathcal{D}_X$. Let 
$$ \ket{\psi_1}=\sum_{x\in X}\sqrt{f_1(x)}\ket{x}\qquad\text{and}\qquad  \ket{\psi_2}=\sum_{x\in X}\sqrt{f_2(x)}\ket{x}\;.$$
 Then 
 $$\T{\ket{\psi_1}\bra{\psi_1}}{\ket{\psi_2}\bra{\psi_2}}\,=\, \sqrt{ 1 - (1-H^2(f_1,f_2))^2}\;.$$
\end{lem}
We require the following definition, which is analogous to Definition \ref{def:compinddist}:
\begin{deff}{\textbf{Computational Indistinguishability of Quantum States}}\label{def:compind}
Two families of density matrices $\{\rho_{0,\lambda}\}_{\lambda\in\mN}$ and $\{\rho_{1,\lambda}\}_{\lambda\in\mN}$ (indexed by the security parameter $\lambda$) are computationally indistinguishable if for all efficiently computable CPTP maps $\mathcal{S}$ there exists a negligible function $\mu(\cdot)$ such that for all $\lambda\in\mN$:
\begin{equation}
\Big| \tr((\ket{0}\bra{0}\otimes \mcI)\mathcal{S}(\rho_{0,\lambda} - \rho_{1,\lambda})\Big| \,\leq\, \mu(\lambda)\;.
\end{equation}
\end{deff}

\subsection{Homomorphic Encryption}\label{sec:homomorphicprelim}
The following definitions are modified versions of definitions from \cite{brakerski2011} and \cite{brakerski2013}. A homomorphic (public-key) encryption scheme HE = (HE.Keygen, HE.Enc, HE.Dec, HE.Eval) is a quadruple of PPT algoritms which operate as follow:
\begin{itemize}
\item \textbf{Key Generation.} The algorithm $(pk,evk,sk)\leftarrow$ HE.Keygen$(1^{\lambda})$ takes a unary representation of the security parameter and outputs a public key encryption key $pk$, a public evaluation key $evk$ and a secret decryption key $sk$.
\item \textbf{Encryption.} The algorithm $c\leftarrow$ HE.Enc$_{pk}(\mu)$ takes the public key $pk$ and a single bit message $\mu\in\{0,1\}$ and outputs a ciphertext $c$. The notation HE.Enc$_{pk}(\mu;r)$ will be used to represent the encryption of a bit $\mu$ using randomness $r$. 
\item \textbf{Decryption.} The algorithm $\mu^*\leftarrow$ HE.Dec$_{sk}(c)$ takes the secret key $sk$ and a ciphertext $c$ and outputs a message $\mu^*\in\{0,1\}$. 
\item \textbf{Homomorphic Evaluation} The algorithm $c_{f}\leftarrow$ HE.Eval$_{evk}(f,c_1,\ldots,c_l)$ takes the evaluation key $evk$, a function $f:\{0,1\}^l\rightarrow\{0,1\}$ and a set of $l$ ciphertexts $c_1,\ldots,c_l$, and outputs a ciphertext $c_f$. It must be the case that for all $c_1,\ldots,c_l$:
\begin{equation}
\mathrm{HE.Dec}_{sk}(c_f) = f(\mathrm{HE.Dec}_{sk}(c_1),\ldots,\mathrm{HE.Dec}_{sk}(c_l))
\end{equation}
with all but negligible probability in $\lambda$. 
\end{itemize}
A homomorphic encryption scheme is said to be secure if it meets the following notion of semantic security:
\begin{deff}[\bfseries CPA Security]\label{def:cpasecurity}
A scheme \textnormal{HE} is IND-CPA secure if, for any polynomial time adversary $\mathcal{A}$, there exists a negligible function $\mu(\cdot)$ such that 
\begin{equation}
\mathrm{Adv}_{\mathrm{CPA}}[\mathcal{A}] \EqDef  |\Pr[\mathcal{A}(pk,evk,\mathrm{HE.Enc}_{pk}(0)) = 1] - \Pr[\mathcal{A}(pk,evk,\mathrm{HE.Enc}_{pk}(1)) = 1]|=\mu(\lambda)
\end{equation}
where $(pk,evk,sk)\leftarrow$ \textnormal{HE.Keygen}($1^{\lambda}$). 
\end{deff}
Note that the above definitions are restricted to encryptions of classical strings, and to functions with classical input/output. The definitions can be extended to encryption of qubits and functions with quantum input/output in a straightforward manner (see \cite{BrakerskiQFHE} or \cite{bitanskyqma}). For the definition of semantic security for quantum encryptions, see \cite{broadbentjeffery2014}. The scheme presented in this paper does satisfy the extended definitions, though we restrict ourselves to the classical definitions above for simplicity.

We now define two desirable properties of homomorphic encryption schemes:
\begin{deff}[\bfseries Compactness and Full Homomorphism]\label{def:compact}
A homomorphic encryption scheme HE is compact if there exists a polynomial $s$ in $\lambda$ such that the output length of HE.Eval is at most $s$ bits long (regardless of $f$ or the number of inputs). A compact scheme is (pure) fully homomorphic if it can evaluate any efficiently computable boolean function. A compact scheme is leveled fully homomorphic if it takes $1^L$ as additional input in key generation, and can only evaluate depth $L$ Boolean circuits where $L$ is polynomial in $\lambda$. A scheme is quantum pure or leveled fully homomorphic if we refer to quantum circuits rather than boolean circuits. 
\end{deff}

\section{Quantum-Capable Classical Homomorphic Encryption Schemes}\label{sec:quantumcapablerequirements}
As described in Section \ref{sec:quantumcapableoverview}, a classical leveled fully homomorphic encryption scheme is quantum-capable if an \cnot\ can be applied with respect to any ciphertext which occurs during the computation. We begin by formalizing the notion of such a ciphertext. This definition is dependent on the depth parameter $L$ (see Definition \ref{def:compact}). Let $\mathcal{F}_L$ be the set of all functions which can be computed by circuits of depth $L$. Assume for convenience that all such functions have domain $\{0,1\}^l$ and range $\{0,1\}$.
\begin{deff}\label{def:ciphertextduringcomp}
For a classical leveled fully homomorphic encryption scheme HE, let $C_{\textrm{HE}}$ be the set of all ciphertexts which can occur during the computation:
\begin{equation}
    C_{\textrm{HE}} = \{\textrm{HE.Eval}_{evk}(f,\textrm{HE.Enc}_{pk}(\mu_1),\ldots,\textrm{HE.Enc}_{pk}(\mu_l)) |f\in\mathcal{F}_L,\mu_1,\ldots\mu_l\in\{0,1\}\}
\end{equation}
\end{deff}
We now define quantum-capable homomorphic encryption schemes:
\begin{deff}[Quantum-Capable Homomorphic Encryption Schemes]\label{def:quantumcapable}
Let $\lambda$ be the security parameter. Let HE be a classical leveled fully homomorphic encryption scheme. HE is \textit{quantum-capable} if there exists an encryption scheme AltHE such that the following conditions hold for all ciphertexts $c\in C_{\textrm{HE}}$. 
\begin{enumerate}
\item There exists an algorithm HE.Convert which on input $c$ produces an encryption $\hat{c}$ under AltHE, where both $c$ and $\hat{c}$ encrypt the same value. 
\item AltHE allows the XOR operation to be performed homomorphically. Moreover, the homomorphic XOR operation is efficiently invertible using only the public key of AltHE: given the public key of AltHE, a ciphertext $c_0$, and the ciphertext which results from applying the homomorphic XOR to input ciphertexts $c_0$ and $c_1$, it is possible to efficiently recover $c_1$.

\item There exists a distribution $D$ which may depend on the parameters of HE and satisfies the following conditions:
\begin{enumerate}
\item The Hellinger distance between the following two distributions is negligible in $\lambda$: 
\begin{equation}\label{eq:freshdist}
\{\textnormal{AltHE.Enc}_{pk}(\mu;r) | (\mu,r)\xleftarrow{\text{\textdollar}}  D\}
\end{equation}
and
\begin{equation}\label{eq:shiftdist}
\{\textnormal{AltHE.Enc}_{pk}(\mu;r) \oplus_H \hat{c} | (\mu,r) \xleftarrow{\text{\textdollar}} D\}
\end{equation}
where $\oplus_H$ represents the homomorphic XOR operation. 
\item It is possible for a \BQP\ server to create the following superposition given access to the public key $pk$: 
\begin{equation}
\sum_{\mu\in\{0,1\},r} \sqrt{D(\mu,r)}\ket{\mu,r}
\end{equation}
\item Given $y = \textnormal{AltHE.Enc}_{pk}(\mu_0;r_0)$ where $(\mu_0,r_0)$ is sampled from $D$, it must be possible to compute $\mu_0, r_0$ given the secret key and possibly additional trapdoor information (which can be computed as part of the key generation procedure). 
\end{enumerate}
\end{enumerate}
\end{deff}
For convenience, we will assume that AltHE and HE have the same public/secret key; this is the case in the example quantum-capable scheme we provide in Section \ref{sec:quantumcapableexample} and also simplifies Section \ref{sec:extensiontohomomorphic}. However, this assumption is not necessary. 

We prove the following claim, which formalizes the \cnot\ given in Section \ref{sec:descofhiddensum}:
\begin{claim}\label{cl:sumcorrectness}
Let HE be a quantum-capable homomorphic encryption scheme and let $c\in C_{\textrm{HE}}$ be a ciphertext encrypting a bit $s$. Consider a \BQP\ machine with access to $c$ and a state $\ket{\psi}$ on two qubits. The \BQP\ machine can compute a ciphertext $y = $ AltHE.Enc$_{pk}(\mu_0,r_0)$, a string $d$ and a state within negligible trace distance of the following ideal state: 
\begin{equation}
(Z^{d\cdot ((\mu_0,r_0)\oplus (\mu_1,r_1))}\otimes X^{\mu_0})CNOT_{1,2}^s\ket{\psi}
\end{equation}
for ($\oplus_H$ is the homomorphic XOR operation) 
\begin{equation}\label{eq:defmuone}
\textnormal{AltHE.Enc}_{pk}(\mu_0,r_0) = \textnormal{AltHE.Enc}_{pk}(\mu_1,r_1) \oplus_H \textnormal{HE.Convert}(c)
\end{equation}
\end{claim}

\begin{proofof}{\textbf{Claim \ref{cl:sumcorrectness}}}
The server first computes $\hat{c} = \textnormal{HE.Convert}(c)$. He then applies the \cnot, as described in Section \ref{sec:descofhiddensum}. Recall that in Section \ref{sec:descofhiddensum}, we used $f_0$ to denote the encryption function of AltHE and $f_1$ to denote the shift of $f_0$ by the homomorphic XOR (which we denote as $\oplus_H$) of $\hat{c}$. At the stage of \eqref{eq:overviewsuperpositionoverD}, the server holds the following state:
\begin{equation}
    \sum_{a,b,\mu\in\{0,1\},r}\alpha_{ab}\sqrt{D(\mu,r)}\ket{a}\ket{b}\ket{\mu,r}\ket{f_a(\mu,r)} 
\end{equation}
\begin{equation}\label{eq:statindclaim1}  
= \sum_{a,b,\mu\in\{0,1\},r}\alpha_{ab}\sqrt{D(\mu,r)}\ket{a}\ket{b}\ket{\mu,r} \ket{\textrm{AltHE.Enc}_{pk}(\mu;r) \oplus_H a\cdot \hat{c}}
\end{equation}
Fix $a = 1$ and $b$ and consider the resulting state: 
\begin{equation}\label{eq:stateindfixa}
\sum_{\mu_0\in\{0,1\},r_0} \sqrt{D(\mu_1,r_1)}\ket{1,b}\ket{\mu_1,r_1}\ket{\textrm{AltHE.Enc}_{pk}(\mu_0;r_0)}
\end{equation}
where $\mu_1,r_1$ are defined with respect to $\mu_0,r_0$ and $\hat{c}$ as in \eqref{eq:defmuone}. We can now directly apply Lemma \ref{lem:hellingertotrace} with reference to the distributions in \eqref{eq:freshdist} and \eqref{eq:shiftdist} (the latter corresponds to \eqref{eq:stateindfixa}). Since the distributions in \eqref{eq:freshdist} and \eqref{eq:shiftdist} are negligibly close, we obtain that the following state is within negligible trace distance of \eqref{eq:stateindfixa}:
\begin{equation}\label{eq:stateindfixalemma}
\sum_{\mu_0\in\{0,1\},r_0} \sqrt{D(\mu_0,r_0)}\ket{1,b}\ket{\mu_1,r_1}\ket{\textrm{AltHE.Enc}_{pk}(\mu_0;r_0)}
\end{equation}
It follows immediately that the state in \eqref{eq:statindclaim1} is within negligible trace distance of the following state: 
\begin{equation}\label{eq:statindclaim2}
\sum_{\mu_0\in\{0,1\},r_0}\sum_{a,b\in\{0,1\}} \alpha_{ab}\sqrt{D(\mu_0,r_0)}\ket{a,b}\ket{\mu_a,r_a}\ket{\textrm{AltHE.Enc}_{pk}(\mu_0;r_0)}
\end{equation}
Observe that, when measured, the state in \eqref{eq:statindclaim2} collapses exactly to the state in \eqref{eq:statecollapseoverview}. The statement of \Cl{cl:sumcorrectness} follows.

\end{proofof}

\section{Example of a Quantum-Capable Classical Encryption Scheme}\label{sec:quantumcapableexample}
This section is dedicated to showing that the dual of the fully homomorphic encryption scheme in \cite{fhelwe} is quantum-capable. We begin by presenting a scheme called Dual in Section \ref{sec:dualscheme}, which is the dual of the encryption scheme from \cite{regev2005}. In Section \ref{sec:dualhescheme}, we use the framework of \cite{fhelwe} to extend Dual to a scheme called DualHE, which we prove in Theorem \ref{thm:dualhehomomorphic} is a classical leveled fully homomorphic encryption scheme. Finally, in Section \ref{sec:dualquantumcapable} (Theorem \ref{thm:dualhequantumcapable}), we prove that DualHE is quantum-capable . 

We begin by listing our initial parameters. Let $\lambda$ be the security parameter. All other parameters are functions of $\lambda$. Let $q\geq 2$ be a power of 2. Let $n,m\geq 1$ be polynomially bounded functions of $\lambda$, let $N = (m+1)\log q$ and let $\beta_{init}$ be a positive integer such that the following conditions hold: 
\begin{equation}\label{eq:assumptionsinitial}
    \begin{minipage}{0.9\textwidth}
\begin{enumerate}
\item  $m = \Omega(n \log q)$ ,
\item $ 2\sqrt{n} \leq \beta_{init}$
\end{enumerate}
    \end{minipage}
  \end{equation}

\subsection{Dual Encryption Scheme}
We first describe the dual scheme of \cite{regev2005}. This scheme was originally given in \cite{gentry2008}, but the presentation below is taken from Section 5.2.2 in \cite{peikertsurvey}. This scheme will eventually serve as the scheme AltHE in Definition \ref{def:quantumcapable}. 
\begin{scheme}\label{sec:dualscheme}{\textbf{Dual Encryption Scheme \cite{gentry2008}}}
\begin{itemize}
\item Dual.KeyGen: Choose $\*e_{sk}\in \{0,1\}^m$ uniformly at random. Using the procedure \newline
$\GenTrap(1^n,1^m,q)$ from Theorem \ref{thm:trapdoor}, sample a random trapdoor matrix $\*A\in \mZ_q^{m\times n}$, together with the trapdoor information $t_{\*A}$. The secret key is $\*{sk} = (-\*e_{sk},1)\in\mZ_q^{m + 1}$ and the trapdoor is $t_{\*A}$. The public key is $\*A'\in \mZ_q^{(m+ 1)\times n}$, which is the matrix composed of $\*A$ (the first $m$ rows) and $\*A^T\*e_{sk}\bmod q$ (the last row).  
\item Dual.Enc$_{pk}(\mu)$: To encrypt a bit $\mu\in\{0,1\}$, choose $\*s\in\mZ_q^{n}$ uniformly and create $\*e\in \mZ_q^{m + 1}$ by sampling each entry from $D_{\mZ_q,\beta_{init}}$. Output $\*A'\*s + \*e + (0,\ldots,0,\mu\cdot \frac{q}{2})\in\mZ_q^{m + 1}$.
\item Dual.Dec$_{sk}(\*c)$: To decrypt, compute $b' = \*{sk}^T\*c \in \mZ_q$. Output 0 if $b'$ is closer to 0 than to $\frac{q}{2}\bmod q$, otherwise output 1. 

\end{itemize}
\end{scheme}
We make a few observations:
\begin{itemize}
\item For a ciphertext $c$ with error $\*e$ such that $\lVert\*e\rVert < \frac{q}{4\sqrt{m+1}}$, the decryption procedure will operate correctly (since $\*{sk}^T\*A' = \*0$). 
\item The trapdoor $t_{\*A}$ can be used to recover the randomness $\*s,\*e$ from a ciphertext. To see this, note that the first $m$ entries of the ciphertext can be written as $\*A\*s + \*e'$, where $\*e'\in \mZ_q^{m}$. Therefore, the inversion algorithm \Invert\ in Theorem \ref{thm:trapdoor} outputs $\*s,\*e$ on input $\*A\*s + \*e'$ and $t_{\*A}$ as long as $\lVert \*e' \rVert < \frac{q}{C_T\sqrt{n\log q}}$ for $C_T$ the universal constant in Theorem~\ref{thm:trapdoor}. 
\item This scheme is naturally additively homomorphic; adding two ciphertexts encrypting $\mu_0$ and $\mu_1$ results in a ciphertext encrypting $\mu_0\oplus \mu_1$. 
\end{itemize}

\subsection{Leveled Fully Homomorphic Encryption Scheme from Dual}\label{sec:dualheschemesec}
We can extend the scheme Dual into a leveled fully homomorphic encryption scheme DualHE in the same way that the standard LWE scheme from \cite{regev2005} is extended in \cite{fhelwe}. Namely, we map the ciphertexts to matrices and encrypt the bit $\mu$ in a matrix (the key generation procedure remains the same). We begin with the required preliminaries and then describe the scheme DualHE. Next, we prove a property of DualHE which is crucial for quantum capability: the ciphertexts retain the same form throughout the computation. Finally, we show in Theorem \ref{thm:dualhehomomorphic} that for a strengthened version of the parameters in \eqref{eq:assumptionsinitial}, DualHE a leveled fully homomorphic encryption scheme. 

To describe this scheme, we will require two operations used in \cite{fhelwe}. The first is the linear operator $\*G\in \mZ_q^{(m + 1)\times N}$ ($N = (m+1)\log_2 q$), which converts a binary representation back to the original representation in $\mZ_q^{m + 1}$. More precisely, consider the $N$ dimensional vector $\*a = (a_{1,0},\ldots,a_{1,l-1},\ldots,a_{m+1,0},\ldots,a_{m+1,l-1})$, where $l = \log_2 q$. $\*G$ performs the following mapping: 
\begin{equation}
\*G(\*a) =(\sum\limits_{j=0}^{\log_2 q - 1} 2^j\cdot a_{1,j},\ldots,\sum \limits_{j=0}^{ \log_2 q - 1}2^j\cdot a_{m+1,j})
\end{equation}
Observe that $\*G$ is well defined even if $\*a$ is not a 0/1 vector. We will call the non linear inverse operation $G^{-1}$, which converts $\*a\in \mZ_q^{m + 1}$ to its binary representation (a vector in $\mZ_2^{N}$). $G^{-1}$ can also be applied to a matrix by converting each column. Note that $\*GG^{-1}$ is the identity operation. In terms of homomorphic evaluation, we will only consider the NAND gate, since we are only concerned with applying Boolean circuits. The scheme can be extended to arithmetic circuits over $\mZ_q$, as described in further detail in \cite{fhelwe}. 

The description of the scheme given below is derived from talks (\cite{fhetalk}, \cite{fhecommitmenttalk}) describing the scheme in \cite{fhelwe}. It is equivalent to the description given in \cite{fhelwe}, but is more convenient for our purposes. 
\begin{scheme}{\textbf{DualHE: Classical Leveled FHE Scheme from Dual}}\label{sec:dualhescheme}
\begin{itemize}
\item DualHE.KeyGen: This procedure is the same as Dual.KeyGen.
\item DualHE.Enc$_{pk}(\mu)$: To encrypt a bit $\mu\in\{0,1\}$, choose $\*S\in\mZ_q^{n\times N}$ uniformly at random and create $\*E\in\mZ_q^{(m + 1)\times N}$ by sampling each entry from $D_{\mZ_q,\beta_{init}}$. Output $\*A'\*S + \*E + \mu \*G\in \mZ_q^{(m + 1)\times N}$.
\item DualHE.Eval$(\*C_0,\*C_1)$: To apply the NAND gate, on input $\*C_0,\*C_1$ output $\*G - \*C_0\cdot G^{-1}(\*C_1)$. 
\end{itemize}
For quantum capability, we will also require the following algorithm, which converts a ciphertext under DualHE to a ciphertext under Dual:
\begin{itemize}
\item DualHE.Convert($\*C$): Output column $N$ of $\*C$
\end{itemize}
Given this algorithm, we can state decryption in terms of Dual\footnote{This is equivalent to the decryption algorithm of \cite{fhelwe}, which is as follows. Let $\*u = (0,\ldots,0,1)\in \mZ_q^{m + 1}$. To decrypt, compute $b' = \*{sk}^T\*C G^{-1}(\frac{q}{2}\*u)$. Output 0 if $b'$ is closer to 0 than to $\frac{q}{2}\bmod q$, otherwise output 1.}:
\begin{itemize}
\item DualHE.Dec$_{sk}(\*C)$: Output Dual.Dec$_{sk}(\textrm{DualHE.Convert}(\*C))$
\end{itemize}
\end{scheme}

\subsubsection{Ciphertext Form}\label{sec:ciphertextform}
We will rely on the fact that, throughout the computation of a Boolean circuit of depth $L$, a ciphertext encrypting a bit $\mu$ can be written in the following form:
\begin{equation}
\*A'\*S + \*E + \mu \*G
\end{equation}
where $\lVert \*E\rVert_{\infty} \leq \beta_{init}(N+1)^L$. This is clearly the structure of the ciphertext immediately after encryption and we now show that this ciphertext structure is maintained after the NAND operation. The NAND operation is performed by computing:
\begin{equation}
\*G - \*C_0\cdot G^{-1}(\*C_1)
\end{equation}
Assume the ciphertexts we begin with are $\*C_b = \*A'\*S_b + \*E_b + \mu_b\*G$ for $b\in \{0,1\}$. Using the fact that $\*G G^{-1}$ is the identity operation, it is easy to see that the result of the NAND operation is: 
\begin{eqnarray}
\*A'\*S' + \*E' + (1- \mu_0\mu_1)\*G
\end{eqnarray}
for
\begin{eqnarray}
\*S' &=& -\*S_0\cdot G^{-1}(\*C_1) - \mu_0 \*S_1\\
\*E' &=& - \*E_0\cdot G^{-1}(\*C_1) - \mu_0 \*E_1
\end{eqnarray}
Note that if both $\lVert \*E_0 \rVert_{\infty}$ and $\lVert \*E_1 \rVert_{\infty}$ are at most $\beta$, then $\lVert \*E' \rVert_{\infty} \leq \beta(N+1)$. It follows that if the scheme DualHE is used to compute a circuit of depth $L$, $\lVert \*E \rVert_{\infty} \leq \beta_{init}(N+1)^L$ for all ciphertexts throughout the computation.

\subsubsection{Encryption Conversion}\label{sec:encryptionconversion}
We now use the above property to prove the correctness of DualHE.Convert. Assume we begin with a ciphertext $\*C = \*A'\*S + \*E + \mu \*G$ under DualHE. The ciphertext $\*c$ (under Dual) will be column $N$ of $\*C$. To see why this is correct, first note that all individual columns of $\*A'\*S + \*E$ are of the form $\*A'\*s + \*e$. Second, observe that column $N$ of $\mu \*G$ is equal to $(0,\ldots,0,\mu\cdot \frac{q}{2})$.

\subsubsection{Proof of Correctness and Security of Scheme \ref{sec:dualhescheme}}
The above two sections allow us to easily prove the following theorem:
\begin{thm}\label{thm:dualhehomomorphic}
Let $\lambda$ be the security parameter. There exists a function $\eta_c$ which is logarithmic in $\lambda$ such that DualHE is IND-CPA secure and leveled fully homomorphic under the hardness assumption $\textrm{LWE}_{n,q,D_{\mZ_q,\beta_{init}}}$ if the conditions in \eqref{eq:assumptionsinitial} as well as the following condition are satisfied:
\begin{equation}\label{eq:betadualhehomomorphic}
    \beta_{init}(N+1)^{\eta_c} < \frac{q}{4(m+1)}.
\end{equation}
\end{thm}

\begin{proof}
We prove that DualHE is IND-CPA secure by relying on the hardness of LWE$_{n,m,q,D_{\mZ_q,\beta_{init}}}$ (i.e. the hardness of LWE with a superpolynomial noise ratio), which implies that the ciphertext is computationally indistinguishable from a uniform string. We can use this LWE assumption as long as the public key $\*A'$ is statistically indistinguishable from a uniformly random matrix (see Section \ref{sec:lweprelim}). Since $\*A$ is selected from a distribution which is statistically indistinguishable from the uniform distribution and $m = \Omega(n\log q)$, $\*A'$ is statistically indistinguishable from uniform due to the leftover hash lemma in \cite{leftoverhash} (see \cite{regev2005} or \cite{peikertsurvey} for more details).  

We now show that DualHE is leveled fully homomorphic. From Section \ref{sec:ciphertextform}, it is clear that the evaluation of the NAND operation is correct and that DualHE is compact. Let $\eta_c$ be larger than the depth of the decryption circuit of DualHE, which is logarithmic in $\lambda$. If we show that the decryption procedure operates correctly after evaluation of a circuit of depth $\eta_c$, the standard bootstrapping technique\footnote{A pure fully homomorphic encryption scheme can be obtained by assuming circular security. A leveled fully homomorphic encryption scheme can be obtained by producing a string of public and secret keys and encrypting each secret key under the next public key - see Section 4.1 of \cite{homomorphic}.} of \cite{homomorphic} can be used to turn DualHE into a leveled fully homomorphic encryption scheme. Due to Section \ref{sec:ciphertextform}, we can assume a ciphertext resulting from a circuit of depth $\eta_c$ can be written as $\*A'\*S + \*E + \mu \*G$, where $\lVert \*E \rVert_{\infty} \leq \beta_{init}(N+1)^{\eta_c}$. It is easy to check that the decryption procedure operates correctly as long as 
\begin{equation}\label{eq:initialcondition0}
\lVert \*E \rVert_{\infty}< \frac{q}{4(m+1)}
\end{equation} 
The condition in \eqref{eq:initialcondition0} is implied by the condition in \eqref{eq:betadualhehomomorphic}.

\end{proof}

\subsection{Quantum Capability of DualHE}\label{sec:dualquantumcapable}
We now prove the following theorem:
\begin{thm}\label{thm:dualhequantumcapable}

Let $\lambda$ be the security parameter, let $\eta_c$ be the logarithmic function in Theorem \ref{thm:dualhehomomorphic}, and let $\eta$ be an arbitrary logarithmic function in $\lambda$. Assume the choice of parameters satisfies the conditions in \eqref{eq:assumptionsinitial} as well as the following condition:
\begin{equation}\label{eq:betaquantumcapable}
    \beta_{init}(N+1)^{\eta+\eta_c} < \frac{q}{4(m+1)}.
\end{equation}
Under the hardness assumption of $\textrm{LWE}_{n,q,D_{\mZ_q,\beta_{init}}}$, the scheme DualHE is quantum-capable. 
\end{thm}
Observe that the only change in parameters involved in making DualHE quantum-capable is increasing the circuit depth by an additive logarithmic factor; this does not change the underlying computational assumption of the hardness of learning with errors with a superpolynomial noise ratio.  

\begin{proof}
To prove Theorem \ref{thm:dualhequantumcapable}, we begin by noting that DualHE is leveled fully homomorphic by Theorem \ref{thm:dualhehomomorphic}. We now show that DualHE is quantum-capable, by listing the requirements for quantum capability and proving that each holds. The scheme corresponding to AltHE will be Dual (Section \ref{sec:dualscheme}). Recall the definition of $C_{\textrm{DualHE}}$ from Definition \ref{def:ciphertextduringcomp}. For all ciphertexts $c\in C_{\textrm{DualHE}}$: 
\begin{enumerate}
\item \textit{There exists an algorithm DualHE.Convert$_{pk}$ which on input $c$ produces an encryption $\hat{c}$ under Dual, where both $c$ and $\hat{c}$ encrypt the same value. }

See Section \ref{sec:encryptionconversion}. 
\item \textit{Dual allows the XOR operation to be performed homomorphically. Moreover, the homomorphic XOR operation is efficiently invertible using only the public key of Dual: given the public key of Dual, a ciphertext $c_0$, and the ciphertext which results from applying the homomorphic XOR to input ciphertexts $c_0$ and $c_1$, it is possible to efficiently recover $c_1$.}

See Section \ref{sec:dualscheme}.
\item \textit{There exists a distribution $D$ which satisfies the following conditions:}
\begin{enumerate}

\item \textit{The Hellinger distance between the following two distributions is negligible in $\lambda$: 
\begin{equation}\label{eq:qcreq0}
\{\textnormal{Dual.Enc}_{pk}(\mu;r)|(\mu,r)\xleftarrow{\text{\textdollar}}  D\}
\end{equation}
and
\begin{equation}\label{eq:qcreq1}
\{\textnormal{Dual.Enc}_{pk}(\mu;r) \oplus_H \hat{c}|(\mu,r) \xleftarrow{\text{\textdollar}} D\}
\end{equation}
where $\oplus_H$ represents the homomorphic XOR operation. }

Let
\begin{equation}\label{eq:defbetaf}
    \beta_f = \beta_{init}(N+1)^{\eta_c+\eta}
\end{equation}
The distribution $D$ will sample $\mu, \*s$ uniformly at random and will sample $\*e$ from the discrete Gaussian distribution $D_{\mZ_q^{m + 1},\beta_{f}}$. Assume that $\hat{c} = \textnormal{DualHE.Convert}(c) = \*A'\*s' + \*e' + (0,\ldots,0,s\cdot \frac{q}{2})$ where $\lVert \*e' \rVert\leq \beta_{init}(N+1)^{\eta_c}\sqrt{m + 1}$. We can assume this since we know the format of the ciphertext throughout the computation (see Section \ref{sec:ciphertextform}). The two distributions corresponding to \eqref{eq:qcreq0} and \eqref{eq:qcreq1} are as follows:
\begin{equation}\label{eq:freshdistproof}
\{\*A'\*s + \*e + (0,\ldots,0,\mu\cdot \frac{q}{2})|(\mu,\*s,\*e)\xleftarrow{\text{\textdollar}}  D\}
\end{equation}
and
\begin{equation}\label{eq:shiftdistproof}
\{\*A'(\*s+\*s') + \*e + \*e' + (0,\ldots,0,(\mu\oplus s)\cdot \frac{q}{2})|(\mu,\*s,\*e) \xleftarrow{\text{\textdollar}} D\}
\end{equation}
The Hellinger distance between the distributions in \eqref{eq:freshdistproof} and \eqref{eq:shiftdistproof} is equal to the distance between the following two distributions: 
\begin{equation}\label{eq:gaussiandist}
\{\*e|\*e\xleftarrow{\text{\textdollar}}  D_{\mZ_q^{m + 1},\beta_{f}}\}
\end{equation}
and
\begin{equation}\label{eq:gaussiandistshift}
\{\*e + \*e'|\*e \xleftarrow{\text{\textdollar}} D_{\mZ_q^{m + 1},\beta_{f}}\}
\end{equation}
Since $\lVert \*e' \rVert\leq \beta_{init}(N+1)^{\eta_c}\sqrt{m + 1}$ and $\frac{\beta_f}{\beta_{init}(N+1)^{\eta_c}}$ is equal to the superpolynomial function $(N+1)^{\eta}$, Lemma \ref{lem:distributiondistance} shows that the distance between the two distributions is negligible. 

\item \textit{It is possible for a \BQP\ server to create the following superposition:} 
\begin{equation}
\sum_{\mu\in\{0,1\},r} \sqrt{D(\mu,r)}\ket{\mu,r}
\end{equation}

In this case, $r = (\*s,\*e)$. $D$ samples $\mu$ and $\*s$ according to the uniform distribution and samples $\*e$ according  to the discrete Gaussian distribution $D_{\mZ_q^{m + 1},\beta_{f}}$. It is easy for a \BQP\ server to create a superposition over a discrete Gaussian (see Lemma 3.12 in \cite{regev2005}\footnote{Taken from \cite{oneproverrandomness} - specifically, the state can be created using a technique by Grover and Rudolph (\cite{distributionsuperpositions}), who show that in order to create such a state, it suffices to have the ability to efficiently compute the sum $\sum\limits_{x=c}^d D_{\mZ_q,B_P}(x)$  for any $c,d\in\{-\lfloor\sqrt{B_P}\rfloor,\ldots,\lceil\sqrt{B_P}\rceil\}\subseteq \mZ_q$  and to within good precision. This can be done using standard techniques used in sampling from the normal distribution.}). 

\item \textit{Given $y = \textnormal{AltHE.Enc}_{pk}(\mu_0;r_0)$ where $(\mu_0,r_0)$ is sampled from $D$, it must be possible to compute $\mu_0, r_0$ given the secret key and possibly additional trapdoor information (which can be computed as part of the key generation procedure).}

We first show that $\mu_0,r_0$ can be recovered from $y$. Assume $y$ has error $\*e\in\mZ_q^{m+1}$. Since $\*e$ is sampled from $D_{\mZ_q^{m + 1},\beta_{f}}$, $\lVert \*e \rVert \leq \sqrt{m+1}\beta_f$. Therefore, it is possible to compute $\mu_0$ as long as $\beta_f < \frac{q}{4(m+1)}$ (see Section \ref{sec:dualscheme}). Second, it is possible to recover the randomness $r_0$ of $y$ as long as the lattice trapdoor is applicable. As stated in Theorem \ref{thm:trapdoor}, the lattice trapdoor is applicable if $\beta_f < \frac{q}{C_T\sqrt{n(m+1)\log q}}$. Combining these two conditions, we require that:
\begin{equation}
    \beta_f < \min(\frac{q}{C_T\sqrt{n(m+1)\log q}},\frac{q}{4(m+1)}) = \frac{q}{4(m+1)}
\end{equation}
The equality follows since $m = \Omega(n\log q)$. Given the definition of $\beta_f$ in \eqref{eq:defbetaf}, this condition is satisfied by \eqref{eq:betaquantumcapable}.

\end{enumerate}
\end{enumerate}
\end{proof}

\section{Extension to Quantum Leveled Fully Homomorphic Encryption}\label{sec:extensiontohomomorphic}
We now present the construction of a quantum leveled fully homomorphic encryption scheme from a quantum-capable classical leveled fully homomorphic encryption scheme, as described in Section \ref{sec:extensiontohomomorphicoverview}. We first provide a full description of the scheme (Section \ref{sec:schemedescription}) and then proceed to proving that it is a quantum leveled FHE scheme. Correctness of evaluation follows almost immediately from the correctness of the \cnot\ (see Claim \ref{cl:sumcorrectness}), while CPA security follows along the lines described in Section \ref{sec:extensiontohomomorphicoverview}. 

Assume there are $L$ levels of the quantum circuit to be computed, where each level consists of Clifford gates, followed by a layer of non intersecting Toffoli gates. Note that this circuit arrangement increases the depth of the original circuit by a factor of at most 2. Let the depth of the classical circuit corresponding to each level of this circuit be $L_c$ (this includes decrypting and recovering randomness from ciphertexts corresponding to the \cnot s from the previous level, performing the Pauli key updates corresponding to the \cnot s from the previous level, and performing the Pauli key updates corresponding to the Clifford and Toffoli gates of the current level). See Section \ref{sec:extensiontohomomorphicoverview} for a reminder of the above description. 

The scheme is quite straightforward: the server initially receives a quantum standard basis state encrypted under Pauli keys (which can be thought of as a one time padded classical string), along with the Pauli keys encrypted under a quantum-capable classical homomorphic encryption scheme (which we call HE) and a string of evaluation keys (one for each level). Each evaluation key consists of the evaluation key of HE, encrypted secret key/ trapdoor information and a fresh public key for each level. Recall that the evaluation key is of this form since we need to use a fresh public/ secret key pair for each \cnot; this allows encryption of the secret key/ trapdoor of each level under a new, independent public key (see Section \ref{sec:extensiontohomomorphicoverview}). The server then applies Toffoli and Clifford gates (which compose a universal gate set) as described in Section \ref{sec:overviewhomomorphic}. Finally, the decryption consists of the client first decrypting the encryptions of the Pauli keys, and then using the Pauli keys to decrypt the final measurement result sent by the server.

\begin{scheme}\label{sec:schemedescription}{\textbf{Quantum Leveled Fully Homomorphic Encryption}}
Let HE be a classical leveled fully homomorphic encryption scheme which is quantum-capable for depth $L_c$. 
\begin{itemize}
\item QHE.KeyGen($1^{\lambda}, 1^L$):
\begin{enumerate}
\item For $1\leq i\leq L + 1$, let $(pk_i,evk_i,sk_i, t_{sk_i}) = $HE.Keygen($1^{\lambda}, 1^{L_c}$), where $t_{sk_i}$ is the trapdoor information required for randomness recovery from ciphertexts. 
\item The public key $pk$ is $pk_1$ and the secret key $sk$ is $sk_{L+1}$. The evaluation key $evk$ consists of $(evk_1,\ldots,evk_{L+1})$ as well as $(pk_{i+1},$HE.Enc$_{pk_{i+1}}(sk_{i})$, HE.Enc$_{pk_{i+1}}(t_{sk_i})$) for $1\leq i\leq L$. 
\end{enumerate}
\item QHE.Enc$_{pk}(m)$: For a message $m\in\{0,1\}^\lambda$, the encryption is ($Z^zX^x\ket{m},\textnormal{HE.Enc}_{pk_1}(z,x))$\footnote{Observe that this encryption can immediately be extended to quantum states by replacing $m$ with a $\lambda$ qubit state $\ket{\psi}$. The decryption can also be extended in the same manner.}, where $z,x\in\{0,1\}^{\lambda}$ are chosen at random. Note that $Z^zX^x\ket{m}$ can be represented as the classical string $x\oplus m$. 
\item QHE.Dec$_{sk}$: The input is a classical message $m\in\{0,1\}^{\lambda}$ and encryptions of $z,x\in\{0,1\}^{\lambda}$ under $pk_{L+1}$. The encryptions are first decrypted using $sk_{L+1}$ to obtain $z,x$. The decrypted message is $Z^zX^x\ket{m}$, which can be represented as $x\oplus m$. 
\item QHE.Eval: Clifford gates and Toffoli gates are applied to an encrypted state as follows:
 \begin{enumerate}
\item To apply a Clifford gate, the Clifford is applied to the Pauli one time padded input state and the encrypted Pauli keys are homomorphically updated according to which Clifford gate was applied.
\item To apply a Toffoli gate:
\begin{enumerate}
\item The Toffoli gate is applied to the Pauli one time padded state. Assume the Toffoli is applied on top of the Pauli one time pad $Z^zX^x\in\mbP_3$. 
\item The Pauli key encryptions are homomorphically updated  according to $P_{zx}$.
\item Three \cnot s are used to correct $C_{zx}$ (see Section \ref{sec:toffoliapp} for details on $C_{zx}$ and $P_{zx}$). As part of each operation, the Pauli key encryptions are homomorphically updated (see Claim \ref{cl:leveledsumcorrectness} for a full description of how this is done). 
\end{enumerate}
\end{enumerate}
\end{itemize} 
\end{scheme}

\subsection{CPA Security}\label{sec:quantcpa}
In this section, we prove the following theorem:
\begin{thm}\label{thm:quantcpa}
The scheme presented in Section \ref{sec:schemedescription} is IND-CPA secure. 
\end{thm}
\begin{proof}
To prove CPA security as defined in Definition \ref{def:cpasecurity}, we show that for any polynomial time adversary $\mathcal{A}$, there exists a negligible function $\mu(\cdot)$ such that 
\begin{equation}\label{eq:quantcpa}
\mathrm{Adv}_{\mathrm{CPA}}[\mathcal{A}] =  |\Pr[\mathcal{A}(pk,evk,\mathrm{QHE.Enc}_{pk}(0)) = 1] - \Pr[\mathcal{A}(pk,evk,\mathrm{QHE.Enc}_{pk}(1)) = 1]|=\mu(\lambda)
\end{equation}
where $(pk,evk,sk)\leftarrow$ \textnormal{QHE.Keygen}($1^{\lambda}$). 

The only difficulty in proving \eqref{eq:quantcpa} is that encryptions of $sk_1$ and $t_{sk_1}$ are also given to the attacker as part of the evaluation key; we need to prove that this information can be replaced with encryptions of 0. This can be done via standard techniques in proving security of leveled homomorphic encryption schemes (see Section 4.1 in \cite{homomorphic}). We include the proof for completeness.

We proceed through $L$ hybrids. In the final hybrid, the attacker is given only the public key $pk_1$ and the evaluation key $evk'$, which consists of $(evk_1,pk_2,evk_2,\ldots, pk_{L+1},evk_{L+1})$ and 2 encryptions of 0 under $pk_{i+1}$ (i.e. HE.Enc$_{pk_{i+1}}(0)$) for $1\leq i\leq L$. CPA security at this point follows immediately (by replacing the encryptions of $z,x$ with 0 and then using Lemma \ref{paulimix}). The hybrids are as follows:
\begin{description}
\item Hyb$_{L+1}$: The evaluation key is as described in Section \ref{sec:schemedescription}.
\item For $1\leq i\leq L$, where $i$ is decreasing:
\begin{description}
\item Hyb$_i$: The evaluation key is the same as in Hyb$_{i+1}$, except HE.Enc$_{pk_{i+1}}(t_{sk_{i}})$ and HE.Enc$_{pk_{i+1}}(sk_{i})$ are replaced with encryptions of 0. 
\end{description}
\end{description}
Note that in Hyb$_i$, the evaluation key does not contain secret key or trapdoor information corresponding to public keys $pk_i,\ldots,pk_{L+1}$. More specifically, the evaluation key in Hyb$_1$ is $evk'$ (all the encryptions of secret keys and trapdoors have been replaced by encryptions of 0). 

First, Hyb$_{L+1}$ is computationally indistinguishable from Hyb$_L$ due to the CPA security of HE under $pk_{L+1}$ (note that encryptions of $sk_{L+1}$ and $t_{sk_{L+1}}$ were not provided as part of the evaluation key). For all $1\leq i\leq L-1$, Hyb$_{i+1}$ is indistinguishable from Hyb$_{i}$ due to the CPA security of HE under $pk_{i+1}$. This is because Hyb$_{i+1}$ has no secret key or trapdoor information corresponding to $pk_{i+1}$. 

It follows that there exists a negligible function $\mu_C$ such that the CPA security of QHE 
\begin{eqnarray}
|\Pr[\mathcal{A}(pk,evk,\mathrm{QHE.Enc}_{pk}(0)) = 1] - \Pr[\mathcal{A}(pk,evk,\mathrm{QHE.Enc}_{pk}(1)) = 1]|
\end{eqnarray}
can be upper bounded as follows:
\begin{eqnarray}
\cdots &\leq& \mu_C L + |\Pr[\mathcal{A}(pk,evk',\mathrm{QHE.Enc}_{pk}(0)) = 1] - \Pr[\mathcal{A}(pk,evk',\mathrm{QHE.Enc}_{pk}(1)) = 1]|\nonumber\\
&\leq& \mu_C (L+1)
\end{eqnarray}
where $(pk,evk,sk)\leftarrow$ \textnormal{QHE.Keygen}($1^{\lambda}$) and $evk' = (evk_1,pk_2,evk_2,\ldots, pk_{L+1},evk_{L+1})$. 

\end{proof}

\subsection{Quantum Leveled FHE}
In this section, we prove the following theorem:
\begin{thm}\label{thm:qhehomomorphic}
The scheme QHE presented in Section \ref{sec:schemedescription} is a quantum leveled fully homomorphic encryption scheme.   
\end{thm}
Combining Theorem \ref{thm:qhehomomorphic} with Theorem \ref{thm:dualhequantumcapable} provides the main result of this paper (stated informally in Theorem \ref{thm:mainresult}). To prove Theorem \ref{thm:qhehomomorphic}, we need to prove that QHE can evaluate depth $L$ quantum circuits. This is taken care of by the following claim:
\begin{claim}\label{cl:leveledsumcorrectness}
Assume the underlying classical encryption scheme HE of QHE is quantum-capable for depth $L_c$. Then a \BQP\ machine with access to a ciphertext $c$ encrypting $s$ under $pk_i$, a quantum state $Z^zX^x\ket{\psi}$ on two qubits, ciphertexts encrypting $z,x$ under $pk_i$, and the evaluation key of QHE can compute the encryptions of $z',x'\in\{0,1\}^2$ under $pk_{i+1}$ as well as a state within negligible trace distance of the following ideal state
\begin{equation}
\textrm{CNOT }_{1,2}^sZ^{z'}X^{x'}\ket{\psi}\bra{\psi}(Z^{z'}X^{x'})^\dagger(\textrm{CNOT }_{1,2}^s)^\dagger
\end{equation}
\end{claim}
\begin{proof}
Let $c_{z,x,pk_i}$ be the concatenation of four ciphertexts, each encrypting a single bit of $z,x$ under $pk_i$. The server applies the following operations:
\begin{enumerate}
\item As described in Section \ref{sec:descofhiddensum}, the server applies the encrypted CNOT operation to the two-qubit state $Z^zX^x\ket{\psi}$ using the ciphertext $\hat{c} = $HE.Convert$(c)$. According to Claim \ref{cl:sumcorrectness}, the server will obtain a ciphertext 

$y = $AltHE.Enc$_{pk}(\mu_0,r_0)$, a string $d\in\{0,1\}^m$ and a state within negligible trace distance of the following ideal state:
\begin{equation}
(Z^{d\cdot ((\mu_0,r_0)\oplus (\mu_1,r_1))}\otimes X^{\mu_0})\textrm{CNOT}_{1,2}^s\ket{\psi}
\end{equation}
where $\textnormal{AltHE.Enc}_{pk}(\mu_0;r_0) = \textnormal{AltHE.Enc}_{pk}(\mu_1;r_1) \oplus_H \hat{c}$ and $\oplus_H$ is the homomorphic XOR operation.

\item The server uses $pk_{i+1}$ to compute HE.Enc$_{pk_{i+1}}(c_{z,x,pk_i})$ and HE.Enc$_{pk_{i+1}}(\hat{c},y,d)$. 
\item The server computes the encryption of $z,x$ under $pk_{i+1}$ by homomorphically running the decryption circuit on inputs $\mathrm{HE.Enc}_{pk_{i+1}}(sk_i)$ and HE.Enc$_{pk_{i+1}}(c_{z,x,pk_i})$ .
\item The server homomorphically computes $(\mu_0,r_0)$ and $(\mu_1,r_1)$, using the ciphertexts encrypting $t_{sk_i},sk_i,\hat{c},y,d$ (all encrypted with HE under public key $pk_{i+1}$). The server then uses this result, along with the ciphertexts encrypting $z,x,d$, to homomorphically compute $z' = z + (d\cdot ((\mu_0,r_0)\oplus (\mu_1,r_1)),0)$ and $x' = x + (0,\mu_0)$. The result of this computation is the encryption of $z',x'$ with HE under $pk_{i+1}$. 
\end{enumerate}
\end{proof}

Claim \ref{cl:leveledsumcorrectness} describes how to perform a single encrypted CNOT operation. However, as stated in 2(c) of Scheme \ref{sec:schemedescription}, three such operations (as described by the operator $C_{zx}$ in Section \ref{sec:toffoliapp}) are required for each Toffoli gate. Since these three operators commute, they can be applied in parallel: Step 1 in Claim \ref{cl:leveledsumcorrectness} can be applied for all three operators, and then Steps 2 - 4 in Claim \ref{cl:leveledsumcorrectness} can be applied for all three operators, thereby implying that only one key switch (from $pk_i$ to $pk_{i+1}$) is needed to apply all three encrypted CNOT operations.

Theorem \ref{thm:qhehomomorphic} follows from Theorem \ref{thm:quantcpa} and Claim \ref{cl:leveledsumcorrectness}:

\begin{proofof}{ \textbf{Theorem \ref{thm:qhehomomorphic}}}
Theorem \ref{thm:quantcpa} shows QHE is IND-CPA secure. From the description of the scheme (Scheme \ref{sec:schemedescription}), it is clear that QHE is compact. Since the number of Toffoli gates is polynomial in $\lambda$, \Cl{cl:leveledsumcorrectness} (along with the triangle inequality) implies that the trace distance of the server's final state from the ideal (correct) final state is at most negligible in $\lambda$. 

\end{proofof}

\section{Acknowledgments}
Thanks to Dorit Aharonov, Zvika Brakerski, Sanjam Garg, Stacey Jeffery, Zeph Landau, Umesh Vazirani and Thomas Vidick for many useful discussions.

\bibliographystyle{alpha}
\bibliography{qpip}

\end{document}